\PassOptionsToPackage{longnamesfirst, authoryear}{natbib}
\PassOptionsToPackage{table}{xcolor}
\documentclass[acmsmall,nonacm]{acmart} %
\usepackage{graphicx} %

\usepackage{amsmath}
\usepackage{listings}
\usepackage{subcaption}
\usepackage{dsfont}

\usepackage{amsthm}

\usepackage{aliascnt}
\usepackage{hyperref}
\usepackage[capitalize,noabbrev]{cleveref}

\usepackage{xcolor}
\usepackage{algorithm}
\usepackage{algpseudocode} %

\newtheorem{theorem}{Theorem}[section]

\newaliascnt{lemma}{theorem}
\newtheorem{lemma}[lemma]{Lemma}
\aliascntresetthe{lemma}
\crefname{lemma}{lemma}{lemmas}
\Crefname{lemma}{Lemma}{Lemmas}

\newaliascnt{corollary}{theorem}

\aliascntresetthe{corollary}
\crefname{corollary}{corollary}{corollaries}
\Crefname{corollary}{Corollary}{Corollaries}

\newaliascnt{proposition}{theorem}

\aliascntresetthe{proposition}
\crefname{proposition}{proposition}{propositions}
\Crefname{proposition}{Proposition}{Propositions}

\theoremstyle{definition}
\newaliascnt{definition}{theorem}
\newtheorem{definition}[definition]{Definition}
\aliascntresetthe{definition}
\crefname{definition}{definition}{definitions}
\Crefname{definition}{Definition}{Definitions}

\theoremstyle{remark}
\newaliascnt{remark}{theorem}

\aliascntresetthe{remark}
\crefname{remark}{remark}{remarks}
\Crefname{remark}{Remark}{Remarks}

\crefname{theorem}{theorem}{theorems}
\Crefname{theorem}{Theorem}{Theorems}

\crefname{lemma}{lemma}{lemmas}
\Crefname{lemma}{Lemma}{Lemmas}

\crefname{corollary}{corollary}{corollaries}
\Crefname{corollary}{Corollary}{Corollaries}

\crefname{proposition}{proposition}{propositions}
\Crefname{proposition}{Proposition}{Propositions}

\crefname{definition}{definition}{definitions}
\Crefname{definition}{Definition}{Definitions}

\crefname{remark}{remark}{remarks}
\Crefname{remark}{Remark}{Remarks}

\usepackage{changepage}
\usepackage{multirow}

\usepackage[most]{tcolorbox}

\newtcolorbox[auto counter, number within=section]{example}[1][]{
  colback=blue!0!white,
  colframe=blue!0!white,
  coltitle=black,
  fonttitle=\bfseries,
  title=Example~\thetcbcounter: #1,
  boxrule=0.8pt,
  arc=4pt,
  outer arc=2pt,
  top=6pt, bottom=6pt, left=6pt, right=6pt,
  enhanced,
}

\usepackage{enumerate}

\usepackage{url}

\usepackage[colorinlistoftodos,prependcaption,textsize=tiny,disable]{todonotes}
\newcommand{\SiSe}[1]{\todo[linecolor=orange,backgroundcolor=orange!25,bordercolor=orange]{{\bf SiSe:} #1}}

\newcommand{\lr}[1]{\left (#1\right)}
\newcommand{\lrs}[1]{\left [#1 \right]}
\newcommand{\lrc}[1]{\left \{#1\right\}}
\newcommand{\lra}[1]{\left |#1\right|}

\renewcommand{\epsilon}{\varepsilon}

\newcommand{\Z}{\mathbb Z}
\newcommand{\R}{\mathbb R}

\newcommand{\MN}[1]{\mathcal N\left( #1 \right)}

\newcommand{\BL}[1]{\Omega\!\left( #1 \right)}
\newcommand{\BU}[1]{\mathcal{O}\!\left( #1 \right)}
\newcommand{\T}[1]{\Theta\!\left( #1 \right)}
\newcommand{\M}{\mathcal{M}}

\NewDocumentCommand{\E}{o}{\mathbb E\IfValueT{#1}{\lrs{#1}}}
\NewDocumentCommand{\1}{o}{\mathds 1{\IfValueT{#1}{\lr{#1}}}}

\newcommand{\norm}[1]{\left\lVert#1\right\rVert}

\title{A Fast Gaussian Mechanism under Continual Observation, with Applications}

\author{Rasmus Pagh}
\email{pagh@di.ku.dk}
\affiliation{%
  \institution{BARC, University of Copenhagen}
  \country{Denmark}
}

\author{Sia Sejer}
\email{sia.sejer@di.ku.dk}
\affiliation{%
  \institution{BARC, University of Copenhagen}
  \country{Denmark}
}

\date{\today}

\begin{document}

\begin{abstract}
    We consider the problem of privately releasing a $k$-dimensional vector under updates:
    Starting with a zero vector, at times $t_1, t_2,\dots$ the vector is updated by adding $x^{(1)}, x^{(2)},\dots$, respectively.
    For positive integers $T$, $k$ we model the updates as a data set $\{(t_i, x^{(i)})\}_i$, where $t_i \in [T]$ and $x^{(i)} \in B_k$ (the $k$-dimensional unit ball).
    Two such data sets are said to be neighboring if their symmetric difference has size at most $1$.
    The continual release consists of the sum $A^{(t)} = \sum_{i \; : \; t_i \leq t} x^{(i)}$ for each time step $t=1,\dots,T$.
    Classical continual release techniques allow us to release an approximation of $A^{(1)},\dots,A^{(T)}$ with additive noise of magnitude $\text{polylog}(T)$, computed in time $\BU{kT}$, even in the on-line, adaptive case where data is continually revealed for the current time step.

    Motivated by private sketching techniques, we consider the setting where only a \emph{subset} of entries in $A^{(t)}$ need to be released at time step $t$.
    Our new result is that it is possible to sample any desired entry in a given noise vector in \emph{constant time} while reproducing exactly the distribution of the binary tree mechanism with Gaussian noise.
    The improvement on the known time bound of $\BU{\log T}$ comes from a new data structure that allows us to sample a new noise value with the correct correlations in constant time using Brownian bridges.
    We present two data management applications, of independent interest, that use our technique in conjunction with differentially private CountSketches:
    1) A dynamic data structure for orthogonal range counting queries with a better privacy/accuracy/space trade-off than previous data structures, and
    2) Join size estimation, where in addition we show improved high-probability bounds. %
\end{abstract}

\ccsdesc[500]{Security and privacy~Privacy protections}

\keywords{Differential Privacy, Continual Observation, Sketching, Range Counting, Join Size Estimation}

\newpage
\setcounter{page}{1}
\maketitle

\section{Introduction}
\emph{Differential privacy} provides a strong information-theoretic guarantee of privacy with a tunable trade-off between privacy and \emph{utility}, measured by the magnitude of noise added to query answers.
It is relevant in settings where data points may reveal sensitive information about individuals.
The standard way of making a single $k$-dimensional vector private is by adding independent noise (e.g., zero-mean Gaussians) to its entries before release.
In the \emph{continual observation} setting, releases are interleaved with bounded-norm updates to the vector, and we want the whole sequence of vectors to be differentially private.
One way to handle this, with good dependence on the number of queries, is to instantiate $k$ versions of the classical binary mechanism of~\citet{Dwork2010ContinalObs}.
This approach yields noise of magnitude polylog$(T)$ over $T$ time steps, but updates all $k$ noise values at each time step and thus requires $\BL{k}$ time per release.
This is optimal if we wish to release the whole vector at each step.
However, updates may be sparse and full releases may be replaced by \emph{queries} that release only certain vector entries.
In this case we can hope to optimize query and update time to depend on the number of distinct times a vector entry is queried and updated.

\medskip

{\bf Sketches under continual observation.}
A canonical example of a sparse update/sparse query setting is private sketching.
Applying continual observation to private linear sketches has previously been considered by~\citet{DifferentiallyPrivateContinualReleasesofStreamingFrequencyMomentEstimations}, but their technique requires updating every sketch entry in each time step, making it inefficient in settings where sketches are large.
\citet{ScalableDifferentiallyPrivateSketchesUnderContinualObservation} developed a technique in which sketch updates are buffered and carried out in batches when the number of updates matches the sketch size.
While this improves time complexity, buffering updates introduces significant error for large sketches.

In this paper we show that such trade-offs are not necessary.
Our main result is the Fast Gaussian Mechanism under Continual Observation (FastGaMe), a time-efficient implementation of the binary tree mechanism of~\cite{Dwork2010ContinalObs} in the Gaussian noise setting.
FastGaMe \emph{only materializes the noise values that are actually queried}, and generates each such noise value in \emph{constant time}.
This implies that any sketch can be implemented with continual observation noise and the same asymptotic time complexity as in the non-private setting.
Like in~\citep{DifferentiallyPrivateContinualReleasesofStreamingFrequencyMomentEstimations,ScalableDifferentiallyPrivateSketchesUnderContinualObservation} the space usage grows by a factor $\Theta(\log T)$, where $T$ is the number of time steps.%

\medskip

{\bf Applications using CountSketch.}
We show two applications of FastGaMe, both based on \emph{Private CountSketch}, recently studied by \citet{PCS} and \citet{DifferentiallyPrivateLinearSketches}.
This sketch approximately represents a sparse vector under single-entry updates and queries, where each update and query only involves a small number of values from the sketch.
The applications are key primitives in data management:
\begin{itemize}

    \item \textbf{Private dynamic orthogonal range counting.} Given a dynamically changing point set in $[B]^d$, we seek to maintain a differentially private index that supports orthogonal range counting queries (points satisfying a conjunction of range predicates).
    Past work has relied on representing sparse vectors using additive noise with thresholding, and has supported dynamic updates only through expensive rebuilding techniques.
    The combination of FastGaMe and Private CountSketch significantly improves the time and space complexity achievable by past techniques while also providing a stronger privacy/utility guarantee (see \Cref{figure:rangecounting-overview-dynamic} for an overview).

    \item \textbf{Continual join size estimation.} Given two dynamically changing relations, we wish to privately maintain an estimate of their join size.
    Maintaining a CountSketch for the join attributes of each relation allows such estimates to be maintained in small space, but to our knowledge this problem has not been studied in the private setting, let alone the continual observation setting.
    We give bounds for the join size estimation accuracy achievable using Private CountSketch and show that the combination with FastGaMe allows such estimates to be maintained in constant time per update with space only slightly larger than in the non-private setting.
\end{itemize}

\subsection{Related Work}

{\bf Continual release of a counter.}
Continual counting is the basic problem of releasing, after each bounded size update of a scalar value, an approximation of the value while preserving the privacy of individual updates.
This problem and its generalization to $k$-dimensional bounded-norm vector updates is one of the canonical examples of privacy under continual observation, and is often used as a primitive in other continual-release algorithms.
The classical approach is the \emph{binary tree mechanism} of \citet*{Dwork2010ContinalObs} and the closely related construction of \citet{Chan2011ReleaseStatistics}, in which each output is represented as the sum of a logarithmic number of noisy dyadic partial sums.
Both approaches give polylogarithmic error over a horizon of length $T$, forming the standard baseline for continual release.
Recent work has improved these results by providing tighter constant factors, for example: \citet{AlmostTightErrorBoundsOnDifferentiallyPrivateContinualCounting} gave nearly tight bounds for continual counting (with noise value generation in time $\BU{T\log T}$), while \citet{SmoothBinaryMechanismForEfficientPrivateContinualObservation} presented a variant that reduced the variance while allowing all $T$ additive noise values to be generated in time $\BU{T}$.

{\bf Private Linear Sketching.}
Linear sketches are a central tool for summarizing high-dimensional frequency vectors under streaming updates, with CountSketch~\citep{CountSketchCCF} being a canonical example that supports unbiased point queries and self-join size estimation from a small randomized linear summary.
Privacy for sketches has been studied both in the pan-private streaming model, where the internal state of the algorithm must remain private against intrusions~\citep{PanPrivateAlgorithmsViaStatisticsOnSketches}, and in settings where a sketch is made differentially private by perturbing its counters (this paper considers the latter).
Recent work on Private CountSketch shows that adding appropriately calibrated noise to the sketch entries yields private sparse-vector representations with strong accuracy guarantees for point queries and downstream estimation tasks~\citep{PCS,DifferentiallyPrivateLinearSketches}.

Related private heavy-hitters methods that are \emph{not} linear sketches include the private Misra--Gries algorithm of \citet{DifferentiallyPrivateMisraGries}, which combines a deterministic frequency summary with privacy-preserving noise, and the differentially private weighted-sampling framework of \citet{DifferentiallyPrivateWeightedSampling}, which provides another compact summary for private vector and frequency estimation.
These methods work well in settings where data can only be inserted, in contrast to private linear sketches which also support deletions, and have not been shown to work in continual observation settings.

\medskip

{\bf Linear Sketching under Continual Observation.}
The closest prior work on private linear sketches under continual observation is due to \citet{DifferentiallyPrivateContinualReleasesofStreamingFrequencyMomentEstimations}, who study the continual release of frequency-moment estimates in insertion-only streams.
Their algorithms combine streaming sketches with continual-release primitives for counts, distinct elements, heavy hitters, and low-frequency elements, and obtain near-optimal space bounds up to polylogarithmic factors.
However, they use the standard way of applying continual observation to the sketch vector, which requires updating the noise for every sketch cell at every time step.
More recently, \citet{ScalableDifferentiallyPrivateSketchesUnderContinualObservation} considered this computational bottleneck directly and proposed a scalable framework for differentially private sketches under continual observation.
Their method groups stream updates into batches and applies lazy sketch updates, reducing the per-update cost relative to maintaining a fully updated noisy sketch at every time step.
This makes continual observation sketches more practical for high-throughput streams, especially for heavy-hitter detection, but the buffering step also creates an additional approximation term that grows with the sketch size, limiting what precision is possible.

\medskip

{\bf Private Orthogonal Range Counting Queries.}
An orthogonal range counting query finds the number of tuples in a database that satisfy range predicates on up to $d$ designated attributes.
Under widely held assumptions, even the decision problem of determining whether or not a range is empty requires near-linear query time when $d=\BL{\log n}$ and updates must be polynomial time~\cite{Afshani2022HierarchicalCategories,Chan2017ModerateDimensions}.
We thus focus our discussion on the more feasible case where $d$ is constant.

 \citet*{Xiao2011Wavelet} and \citet*{Chan2011ReleaseStatistics} independently proposed similar solutions with error polylogarithmic in the number of queries in the \emph{static case} where updates are not allowed.
The size of both data structures is $\BU{B^d}$ which may in general be much larger than the size $n$ of the dataset.
Improved space and error bounds for the ``sparse data'' setting were obtained by \citet*{PureDPRecQueries}.
When $B\gg n$ they use a reduction that transforms points in $[B]^d$ to points in $[n]^d$ at the cost of additive error $\BU{\log(B)/\varepsilon}$, where $\varepsilon$ is the privacy parameter.
Applied to this smaller problem, the techniques of \citep{Xiao2011Wavelet,Chan2011ReleaseStatistics} have query time that grows only logarithmically with $B$ and use space $n^d$.
To further reduce space, they propose a pruning technique that removes all small noisy counts in the multi-dimensional tree, reducing space to $\BU{n (\log n)^d}$ at the cost of increasing error.
Past upper bounds were for $\varepsilon$-differential privacy, but in this paper we focus on the more modern \emph{$\rho$-zero-concentrated differential privacy} (zCDP) notion which results in better error bounds under composition.
To make these bounds more easily comparable, we set $\rho = \varepsilon^2/2$, the zCDP parameter implied by $\varepsilon$-differential privacy.

In the \emph{dynamic} case we wish to provide privacy in the continual observation setting: Start with the empty point set and perform $n$ updates (insertion or deletion of a point) with arbitrarily many queries asked between updates, with a privacy guarantee for each point update over all query answers.
\Cref{figure:rangecounting-overview-dynamic} gives an overview of data structures for dynamic private orthogonal range counting in the grid $[B]^d$.
\citet{DifferentialPrivacyonFullyDynamicStreams} recently presented a general transformation from static to dynamic data structures under continual observation for a wide class of private data structures.
Applying their technique to the data structure of~\citet{Chan2011ReleaseStatistics} and the low-space data structure of~\citet{PureDPRecQueries} yields the results stated in the first two lines of~\Cref{figure:rangecounting-overview-dynamic}.

Using the discrepancy method introduced by \citet{OptimalPrivateHalfspaceCountingViaDiscrepancy}, lower bounds for private range counting can be derived from lower bounds on the hereditary discrepancy of axis-parallel rectangles.
Building on this framework, \citet*{FactorizationNormsAndHereditaryDiscrepancy} showed that rectangle workloads over $n$ points in $\mathbb{R}^d$ and over the grid $[B]^d$ have hereditary discrepancy $\BL{(\log n)^{d-1}}$ and $\Theta((\log B)^{d-1})$, respectively.
Via the discrepancy–privacy reduction, this implies that any $(\varepsilon,\delta)$-DP mechanism with sufficiently small $\delta$ must incur worst-case error at least $\BL{(\log n)^{d-1}/\varepsilon}$ (for some large enough $B$) or $\BL{(\log B)^{d-1}/\varepsilon}$ (for some large enough $n$).
Thus, the error upper bounds in \Cref{figure:rangecounting-overview-dynamic} are close to optimal up to polylogarithmic factors.

\begin{figure}[t]
\makebox[\linewidth]{
    \begin{tabular}{l|c|c|c}
         {\bf Reference} & {\bf Privacy} & {\bf Maximum error} & {\bf Space}\\ \hline\hline
        \citep{Xiao2011Wavelet} + \citep{DifferentialPrivacyonFullyDynamicStreams} &  $\varepsilon$-DP & $(\log B)^{1.5d+1}(\log n)^{3+\eta}/\varepsilon$ & $B^{d} \log(n)$ \\
    \hline
          \citep{PureDPRecQueries} + \citep{DifferentialPrivacyonFullyDynamicStreams} &  $\varepsilon$-DP & $(\log(B)+ (\log n)^{1.5d+2})(\log n)^{3+\eta}/\varepsilon$ & $n (\log n)^{d+1}$ \\
   \hline
          \multirow{2}{*}{\citep{OptimalPrivateHalfspaceCountingViaDiscrepancy} + \citep{FactorizationNormsAndHereditaryDiscrepancy}
        } & \multirow{2}{*}{$(\varepsilon,\tfrac{1}{100})$-DP} & \cellcolor{red!25} $\BL{(\log n)^{d-1}/\varepsilon}$ & \multirow{2}{*}{any}\\
\cline{3-3}
         &  & \cellcolor{red!25}$\BL{(\log B)^{d-1}/\varepsilon}$ & \\
   \hline
          {\bf This paper} &  $\varepsilon^2/2$-zCDP & $E\geq (\log B)^{d+2}\log(n)/\varepsilon$ & $n (\log B)^{d+1} \log(n)/E$ \\
          \hline
    \end{tabular}}
    \caption{Overview of results on dynamic private orthogonal range counting with $n$ insertions/deletions of points in $[B]^d$. %
    Lower bounds hold for worst-case \emph{static} inputs and are expressed in terms of $n$ or $B$ with no restriction on the other parameter.
    They hold under approximate DP and in particular for $\varepsilon^2/2$-zCDP.
    For simplicity, we assume that both $d$ and the privacy parameter $\varepsilon$ are $\BU{1}$, and order-notation is suppressed.
    Maximum error upper bounds hold with probability $1-1/B$.
    All listed static data structures support an orthogonal range counting query in time $\BU{(\log B)^d}$ or $\BU{(\log B)^{d+1}}$.
    After each update, an arbitrary number of queries can be performed with query time $\BU{(\log B)^{d+1}}$.
    }
    \Description{A table comparing privacy guarantees, maximum error, and space bounds for dynamic private orthogonal range counting.}
    \label{figure:rangecounting-overview-dynamic}
\end{figure}

\medskip

{\bf Sketch-based Join Size Estimation.}
Sketch-based join size estimation was first studied in the seminal work of \citet{AGMS99}, who observed that the size of an equi-join can be estimated from linear sketches of the join-attribute frequency vectors.
In this approach, each relation is summarized by a randomized sketch, and the join size is estimated as an inner product of the two sketches.
The linearity of the sketch makes the estimator naturally compatible with dynamic updates, since insertions and deletions only require updating the sketch entries touched by the changed tuple.
This idea has led to a large body of work on faster, smaller, and more accurate join-size sketches, including recent algorithms that improve the space--accuracy tradeoff and support efficient sketch maintenance under turnstile updates~\citep{WDLAS09,Stausholm21,LPT21,HogsgaardKLNS24}.
These works are closest to our non-private baseline: they show that join-size estimates can be maintained quickly under dynamic updates, but do not provide privacy guarantees.
Private join-size estimation has mostly been studied in different privacy models or as part of broader private query-processing systems.
For example, \citet{SketchesBasedJoinSizeEstimationUnderLocalDifferentialPrivacy} study join-size estimation from locally randomized user reports, and PrivateSQL~\citep{kotsogiannis2019privatesql} provides a differentially private SQL engine that can answer relational queries involving join size queries.
These works are complementary to ours: they address private join estimation or private SQL query answering, but not the problem of maintaining a sketch-based join-size estimator under central continual observation with dynamic updates.

\subsection{Definitions}\label{sec:definitions}

An overview of standard privacy-related definitions used in the paper can be found in~\Cref{app: Preliminaries}.

\medskip

{\bf Problem definition and notation.}
For positive integer $c$, we use $[c]$ to denote the set $\lrc{1, \cdots, c}$.
For parameters $p\in\{1,2\}$, $\Delta_p > 0$ and integer $n$, consider a data set $X = \lrc{(t_i, x^{(i)}) : i \in [n]}$, where $t_i \in \Z$, $x^{(i)} \in \R^k$ with $\norm{x^{(i)}}_p \leq \Delta_p$ for all $i$.
Two data sets are $\ell_p$ neighboring if and only if they have symmetric difference of size at most $1$.
In other words, we consider unbounded neighboring datasets with $\ell_p$-bounded updates.
We will almost exclusively consider the more powerful $\ell_2$ updates, but discuss $\ell_1$ updates when comparing to the classical Laplace binary tree mechanism.

We consider the private release of $A^{(t)} = \sum_{i \; : \; t_i \leq t} x^{(i)}$ for $t=1,\dots,T$.
This can be thought of as the continual release of a vector $A$ that is incremented by a vector $x^{(i)}$ at time $t_i$ (with several updates being possible at each time step).
A query at time step $t$ for entry $i \in [k]$ releases $A^{(t)}_i$ plus an additive noise term.
Any number of queries may be performed at each time step $t$, and we seek privacy guarantees that hold without restrictions on queries.
Updates must respect the timing of queries in the sense that once a time step has been queried, updates are only possible \emph{after} that timestamp.
Updates can skip time steps, jumping ahead any number of steps.

\medskip

{\bf CountSketch.}
A CountSketch~\citep{CountSketchCCF} with $t$ repetitions and table size $b$ uses random hash functions $h_1,\dots,h_t : [d] \to [b]$, and random sign functions $s_1,\dots,s_t : [d] \to \{-1,+1\}$.
The CountSketch of $x$, consisting of $t$ ``rows'' $\text{CS}_1(x),\dots,\text{CS}_t(x) \in \R^b$ has entries defined by
\(
\text{CS}_i(x)_j = \sum_{a : h_i(a)=j} s_i(a)x_a
\).
Given a CountSketch \(\text{CS}_i(x)\), an unbiased estimator for $x_a$ is $\text{median}_i(s_i(a)\,\text{CS}_i(x)_{h_i(a)})$.
The basic analysis of the error of this estimator only relies on pairwise independence of hash functions~\citep{CountSketchCCF}, but some more precise analyses assume the hash functions to be fully random~\citep{ImprovedConcentrationBoundsCountSketch}.
In this paper we will mostly assume full randomness, since this aspect is not our main focus, but we believe that all bounds can be realized with hash functions having constant independence.

\medskip

{\bf Private CountSketch.}
A Private CountSketch is a CountSketch where each sketch entry is perturbed with independent Gaussian noise.
That is, it has rows
\(
    \Tilde{\text{CS}}_i(x) = CS_i(x) + \nu_i
\)
where $\nu_i \sim \mathcal{N}\lr{0,\sigma^2 I_b}$.
CountSketch has $\ell_2$ sensitivity $\sqrt{t}$ and therefore Private CountSketch satisfies $\rho$-zero-concentrated differential privacy with $\rho = \frac{t}{2\sigma^2}$.
Its utility was first studied by \cite{DifferentiallyPrivateLinearSketches, PCS}, where the latter gave a precise analysis of the error distribution of single-entry estimates~$\tilde{x}_a$ obtained from Private CountSketch.

\subsection{Technical Overview}\label{sec:technical-overview}

{\bf New Analysis of the Binary Tree Mechanism (\Cref{sec:continual-counting}).}
Our Fast Gaussian Mechanism under Continual Observation (FastGaMe) is based on the binary tree mechanism for continual counting due to~\citet*{Dwork2010ContinalObs}.
Their method assigns independent noise values to the nodes in a complete binary tree and outputs $A^{(t)}$ with noise obtained from summing the noise of all nodes on the path from the root to leaf number $t$.
Unlike~\citep{Dwork2010ContinalObs} we use the binary tree mechanism with \emph{Gaussian} noise rather than Laplace noise, a combination we refer to as $\M_{BT}$.
Some tree-based continual counting mechanisms have an explicit strategy matrix in the factorization (matrix) mechanism framework of~\citet{li2015matrix}, e.g.~\citep{Chan2011ReleaseStatistics}.
However, the strategy matrix of $\M_{BT}$ has not been studied in the literature to the best of our knowledge, so in \Cref{sec:continual-counting} we provide such an analysis.
Surprisingly, this classical approach improves several binary tree-based methods that have been proposed later~\citep{honaker2015efficient, SmoothBinaryMechanismForEfficientPrivateContinualObservation}.

\medskip

{\bf Fast Gaussian Mechanism (\Cref{sec:FastGaMe}).}
We show how to exploit the structure of noise addition in $\M_{BT}$ to efficiently simulate the mechanism when many noise values are skipped and never accessed.
We assume a standard Word RAM supporting constant-time rank/select on $\BU{\log T}$-bit words (see \Cref{app: Preliminaries}).
It is simple to do ``on demand'' sampling of entries of $\M_{BT}$ in logarithmic time by maintaining the noise values on the path to the most recently released leaf.
It is also straightforward to generate the sequence of all noise values in linear time.
Our contribution is a new method for sampling from the same distribution in constant time per value, regardless of how many time steps have passed since the last noise value was released.
The key idea is to maintain a data structure representing the conditional distributions implied by previous noise releases, and show how one can sample from any such conditional distribution in constant time using a combination of word-level parallelism and Brownian bridges.
The latter technique has previously been used in ``lossless multiple release'' settings where several releases at different privacy levels must be correlated to ensure good combined privacy guarantees~\cite{Koufogiannis2017gradual,andersson2025multiple}.

\medskip

{\bf Orthogonal range counting (\Cref{sec:private-ortogonal-range-queries}).} We use a standard dyadic decomposition that conceptually reduces the orthogonal range counting problem to that of representing a collection of sparse vectors containing counts for fixed dyadic ranges.
Instead of relying on classical ``noise and threshold'' techniques \citep{korolova2009releasing,cormode2012differentially,aumuller2022representing} to represent these vectors, we use a collection of Private CountSketches~\citep{PCS,DifferentiallyPrivateLinearSketches}.
This combination, which does not seem to have been explored in the literature, has multiple advantages:
1) The estimate from each sketch is unbiased and independent, so the total error when summing estimates benefits from cancellation of errors,
2) The sketch size can be chosen such that the error from sketching matches the error from privacy noise, significantly reducing space usage, and
3) Since sketches are linear, continual observation techniques including FastGaMe can be applied.
An interesting aspect of our analysis is that it relies on an $\ell_1$ error guarantee rather than the usual $\ell_2$ error guarantee of CountSketch --- we show how this guarantee is a consequence of the residual $\ell_2$ guarantee of~\citep{PCS}.
We stress that the improvement in \Cref{figure:rangecounting-overview-dynamic} comes from algorithmic advances rather than from adapting a relaxed privacy notion, and would persist if all results were cast in the same privacy notion.

\medskip

{\bf Join Size Estimation (\Cref{sec:private-join-size-estimation}).}
Our results on join size estimation in \Cref{sec:private-join-size-estimation} can be seen as a private version of the CountSketch bounds of \citet*{LPT21}.
To generalize this to Private CountSketch, we need to take into account the bilinear error stemming from sketch error multiplied by noise error, as well as the error stemming from multiplying noise terms with each other.
We also extend the analysis of \citep{LPT21} to allow high-probability bounds rather than just variance and moment bounds, even when the number of repetitions is a small constant.
This enables high probability bounds on join size estimates under continual observation.

\section{Binary Tree Mechanism with Gaussian Noise}\label{sec:continual-counting}

We revisit the classical \emph{binary tree mechanism} of \citet*{Dwork2010ContinalObs}. They considered the problem of privately releasing estimates of sums $A^{(t)}$ with an $\ell_1$ neighboring relation ($p=1$) as defined in \Cref{sec:definitions}.
Traditionally this problem was considered in the setting where all timestamps $t_i$ were distinct numbers in $[T]$ (so that the input can be considered an input stream) and for $k=1$, but the privacy guarantees hold even for the more general setting with repeated timestamps and vectors in $\R^k$.
Without loss of generality we assume that $T> 1$ is a power of two.

The binary tree mechanism is defined in terms of a complete, rooted binary tree with $T$ leaves numbered $0,\dots,T-1$.
We associate a noise variable $N_v$ with each node $v$ of the tree, and let $P_v$ denote the set of nodes on the path from the root node to $v$.
The mechanism releases the estimate
\(
\tilde{A}^{(t)} = A^{(t)} + \sum_{v\in P_t} N_v
\).
In the matrix factorization framework of \citep{li2015matrix}, this corresponds to a mechanism with reconstruction matrix $L$ where $L_{t,v} = 1$ if $v\in P_t$ and~$L_{t,v} = 0$ otherwise and unspecified strategy matrix $R$ such that $LR$ is the lower-triangular all-1s counting matrix.
\Cref{fig:binary-tree-factorization-T4} shows such a factorization for $T=4$.

\begin{figure}
\centering
\small
\[
L R =
\begin{pmatrix}
1 & 1 & 0 & 1 & 0 & 0 & 0\\
1 & 1 & 0 & 0 & 1 & 0 & 0\\
1 & 0 & 1 & 0 & 0 & 1 & 0\\
1 & 0 & 1 & 0 & 0 & 0 & 1
\end{pmatrix}
\quad
\begin{pmatrix}
1/2 & 1/2 & 1/2 & 1/2\\
0 & 0 & -1/2 & -1/2\\
1/2 & 1/2 & 0 & 0\\
1/2 & -1/2 & 0 & 0\\
1/2 & 1/2 & 0 & 0\\
0 & 0 & 1/2 & -1/2\\
0 & 0 & 1/2 & 1/2
\end{pmatrix}
\]
\caption{Factorization with reconstruction matrix $L$ and strategy matrix $R$ for the binary tree mechanism with $T=4$. Columns/rows follow the natural breadth-first order, input coordinates are ordered from left to right. %
}
\Description{Two matrices for the binary tree factorization with four leaves. The first matrix has four rows and seven columns and records root-to-leaf path incidence. The second matrix has seven rows and four columns and contains the corresponding signed strategy coefficients.}
\label{fig:binary-tree-factorization-T4}
\end{figure}

\citet{Dwork2010ContinalObs} showed, without presenting a strategy matrix, that if $N_v \sim \text{Lap}((\log_2(T)+1)/\varepsilon)$ the sequence $\tilde{A}^{(1)},\dots,\tilde{A}^{(T)}$ of estimates satisfies $\varepsilon$-differential privacy under the $\ell_1$ neighboring relation.
In turn, this implies that it satisfies $\rho$-zero-concentrated differential privacy with $\rho = \varepsilon^2/2$~\citep{Bun2026ConcentratedDP}.
The variance of each estimate is $\text{Var}[\tilde{A}^{(t)}] = 2(\log_2(T)+1)^3/\varepsilon^2 = (\log_2(T)+1)^3 / \rho$.

In this paper we consider the binary tree mechanism with \emph{Gaussian} noise, that is, $N_v\sim\mathcal{N}(0,\sigma^2)$, under $\rho$-zero-concentrated differential privacy and the $\ell_2$ neighboring relation ($p=2$).
For simplicity we present the argument for scalars ($k=1$) but all arguments remain valid for vectors in $\R^k$.
It is clear that $\text{Var}[\tilde{A}^{(t)}] = (\log_2(T)+1) \sigma^2$, so what remains is to figure out how $\sigma$ should depend on~$\rho$ and $T$.
By describing an explicit strategy matrix (as a linear transformation) we find
that $\sigma^2 \;=\; (\log_2(T)+2)/(8\rho)$ suffices.
It is interesting to compare the resulting maximum variance to that of other classical binary tree-based mechanisms.
For example, the mechanism of \citet{Chan2011ReleaseStatistics} has a less symmetric factorization where the sensitivity of leaves varies between $1$ and $\log T$, meaning that noise must be scaled to the highest sensitivity.
\citet{SmoothBinaryMechanismForEfficientPrivateContinualObservation} showed how a more symmetric method, the \emph{smooth binary mechanism}, achieves smaller error by only using a small subset of the leaves.
Our method matches the asymptotic error of this much more complicated mechanism, and eliminates a lower-order term in the error bound that is significant for small $T$.
We defer the details of the factorization and its analysis to \Cref{app:binary-mechanism-analysis}.

\section{Fast Gaussian Mechanism for Continual Observation}\label{sec:FastGaMe}

For simplicity we consider the case where $k = 1$ such that the maintained vector aggregate satisfies $A^{(t)} \in \R$ at time $t$.
Like in \Cref{sec:continual-counting} we can extend the data structure to arbitrary $k > 1$ by running it independently for each coordinate.

The $i$th update has timestamp $t_i$, where we assume monotonicity, i.e., $t_1 \leq t_2 \leq t_3 \leq \dots$.
Queries are always with respect to the aggregate $A^{(t)}$ at the time $t$ of the latest update.
Recall that once $A^{(t)}$ has been queried, no further updates with timestamp $t$ are allowed.
We next describe how to extend the Gaussian binary tree mechanism with a data structure that allows constant time noise generation.
The idea is to maintain and lazily sample prefix sums along the current ``active root-to-leaf path'':
For a query at time $t \in \lrc{1, \dots, T}$, we define the active path as the path from the root to the leaf labeled $\text{bin}(t-1)$.
When a query is performed at time $t$, the data structure is updated such that it can output noise value $\nu^{(t)}$ having the same joint distribution with previous noise values as in the Gaussian binary tree mechanism.
Our data structure consists of:
\begin{enumerate}
    \item the timestamp $l$ of the last time a query was made to $A$,
    \item a zero-indexed bit vector $w^{(l)} \in \lrc{0,1}^{\log(T)+1}$ where $w^{(l)}_j = 1$ if and only if the prefix sum at level $j$ on the active path has been sampled at time $l$, and
    \item a zero-indexed vector $P^{(l)} \in \R^{\log(T)+1}$ where $P^{(l)}_j$ stores the level-$j$ prefix sum on the active path whenever $w^{(l)}_j = 1$ (these are described as ``valid''; entries with  $w^{(l)}_j = 0$ are ``invalid'' and undefined).
\end{enumerate}

\noindent
Note that $P^{(l)}_{\log(T)}$ exactly matches the noise of the Gaussian binary tree mechanism at step $l$.
For a bit vector $v\in \{0,1\}^*$ and $c\leq |v|$ let $\text{rank}_c(v) = \sum_{i=1}^c v_i$ be the number of 1s in the first $c$ bits of $v$ and let $\text{select}_i(v) = \min \{ c: \text{rank}_c(v)\geq i \}$ be the position of the $i$th $1$ in $v$.
We use $\text{bin}(z)$ to denote the $\log_2(T)$-bit string encoding $z\in\{0,\dots,T-1\}$ in binary.
Based on the above data structure, whose operations are discussed below, we show the following result:
\begin{theorem}%
    For $T>0$ a power of 2, assume that a machine word has size $\BL{\log(T)}$ and rank/select operations on a bit vector of size $\BU{\log(T)}$ are supported in $\BU{1}$ worst-case time.
    For every choice of $\sigma^2 > 0$ there exists a sequence of noise values $\nu^{(1)}, \dots, \nu^{(T)}$ with the following properties:
    \begin{itemize}
        \item For all $t \in [T]$, $\nu^{(t)} \sim \MN{0,(\log(T)+1)\sigma^2}$,
        \item for all $t_1,t_2 \in [T]$, $Cov(\nu^{(t_1)}, \nu^{(t_2)}) = (|m|+1)\sigma^2$ where $m$ is the longest common prefix of $\text{bin}(t_1-1)$ and $\text{bin}(t_2-1)$, and
        \item for every sequence $t_1 \leq t_2 \leq t_3 \leq \dots$, the sequence of values $\nu^{(t_j)}$, $j = 1, 2, 3, \dots$ can be computed in $\BU{1}$ worst-case time per sample using a data structure $\text{DS}$ of $\BU{\log(T)}$ words.
    \end{itemize}
    \label{thm: Constant look-up time}
\end{theorem}

We can assume without loss of generality that $t_1 < t_2 < t_3 < \dots$ since repeated queries for the same noise value are answered by returning $P^{(l)}_{\log(T)}$.
It is easy to see that the noise vector $\nu = (\nu^{(1)}, \dots, \nu^{(T)})$ has the same covariance matrix as that of the Gaussian binary tree mechanism, which has covariance between two noise values determined by the number of shared nodes on the paths to the corresponding leaves.
Since a Gaussian distribution is uniquely determined by its mean (in this case zero) and covariance matrix, the noise vector $\nu$ has the same distribution as in the Gaussian binary tree mechanism.
As argued in \Cref{sec:continual-counting} setting $\sigma^2 \;=\; (\log_2(T)+2)/(8\rho)$ ensures $\rho$-zCDP under updates with $\ell_2$ sensitivity~1, and in general $\sigma^2$ should be scaled according to the squared $\ell_2$ sensitivity.

\medskip

{\bf Data structure and update rule.}
Let \(h=\log_2(T)\), and index the levels of the binary tree by $j \in \lrc{0,1,\ldots,h}$, where level $0$ is the root and level $h$ consists of the leaves.
We use $b(t)=\operatorname{bin}(t-1)$ as the label of the leaf $t$, and let $(b(t)_{\le j}$ denote the prefix of length $j$, corresponding to the ancestor on the $j$th level.
In the Gaussian binary tree mechanism we can use these prefixes to index node noise such that the noise released at time $t$ is $\nu^{(t)} = \sum_{j=0}^{h} \eta_{b(t)_{\le j}}$,
where the variables $\eta_s\sim \mathcal{N}(0,\sigma^2)$ are independent over all $s$.
Equivalently, if
\[
        S_t(j)=\sum_{r=0}^{j}\eta_{b(t)_{\le r}}, \qquad 0\le j\le h,
\]
then the released noise is $\nu^{(t)}=S_t(h)$.
The data structure will store only partial sums along the currently active root-to-leaf path.
After the most recent query time $l$, the state is $ DS=\langle l,w^{(l)},P^{(l)}\rangle$.
If $w^{(l)}_j=1$, then $P^{(l)}_j$ stores the sampled value of $S_l(j)$ on the active path to $b(l)$, and if $w^{(l)}_j=0$, the corresponding entry of $P^{(l)}$ is undefined (and ignored).
The bit vector fits in one machine word, and the vector $P^{(l)}$ uses $\log(T)+1$ words.

We now describe how to answer a query at time $t \ge l$. %
The case $t=l$ was discussed above.
If this is the first query (i.e., $l$ is not set), we sample
\begin{align*}
        P^{(t)}_0\sim \mathcal{N}(0,\sigma^2),
        \qquad
        P^{(t)}_h=P^{(t)}_0+Z,
        \qquad
        Z\sim \mathcal{N}(0,h\sigma^2),
\end{align*}
with $Z$ independent of $P^{(t)}_0$.
Next, we set $w^{(t)}_0=w^{(t)}_h=1$, mark all other entries invalid, store the two sampled values in $P^{(t)}$, and return $P^{(t)}_h$.
(An initial cost of $\BU{\log T}$ for creating $P^{(t)}$ can be avoided by lazy initialization.)
This is exactly the joint distribution of $(S_t(0),S_t(h))$.

It remains to handle the case $t>l$, where $l$ is set.
We create $P^{(t)}$ by updating $P^{(l)}$.
Let $u$ be the longest common prefix of $b(l)$ and $b(t)$, and write $m=|u|$.
The old and new active paths agree up to level $m$ and are disjoint below level $m$.
Thus, the only value from the old path that is needed for the new path is the shared prefix sum $S_l(m)=S_t(m)$.
If $w^{(l)}_m=1$, this shared value is already stored as $P^{(l)}_m$.
Otherwise, let
\begin{align*}
        a = \max\{j < m : w^{(l)}_j=1\},
        \qquad
        c = \min\{j > m : w^{(l)}_j=1\}.
\end{align*}
The sentinel values $w^{(l)}_0=w^{(l)}_h=1$ ensure that such indices exist whenever $w^{(l)}_m=0$.
By the Word RAM assumption, $a$ and $c$ are found in $\BU{1}$ worst-case time using rank/select operations on~$w^{(l)}$.
Conditioned on the already sampled values $P^{(l)}_a=S_l(a)$ and $P^{(l)}_c=S_l(c)$, the missing value $S_l(m)$ is a Gaussian
Brownian-bridge point (similar to the application in~\citep{andersson2025multiple}):
\begin{align*}
        S_l(m) \sim \mathcal{N} \lr{P^{(l)}_a + \tfrac{m-a}{c-a} \lr{P^{(l)}_c - P^{(l)}_a}, \tfrac{(m-a)(c-m)}{c-a}\sigma^2}
\end{align*}
We sample from this conditional distribution, store the value in $P^{(l)}_m$, and set $w^{(l)}_m=1$.
This is the exact conditional distribution of the missing prefix sum in the Gaussian tree.

Below level $m$, the new path enters a subtree that has not appeared on any previous active path.
Because query times are non-decreasing, at the first bit where $b(l)$ and $b(t)$ differ, the old active path goes to the left subtree and the new active path goes to the right subtree (i.e. all earlier queried leaves lie in a left subtree).
Consequently, conditioned on the shared prefix sum $S_t(m)$, the remaining contribution on the new path is independent of the past.
We therefore sample
\begin{align*}
        P^{(t)}_h = P^{(l)}_m + Z',
        \qquad
        Z' \sim \mathcal{N}(0, (h-m) \sigma^2),
\end{align*}
with $Z'$ independent of all previously sampled values, and release $\nu^{(t)} = P^{(t)}_h$.

Lastly, we update the state.
All valid prefix sums at levels $j \le m$ remain valid, since they lie on the shared part of the old and new active paths.
All valid levels $j>m$ from the old path are invalidated, except for level $h$, as it is set to the newly sampled value $P^{(t)}_h$.
This update of $w$ requires only a constant number of word operations, and the procedure samples only a constant number of Gaussian variables.
Thus, each queried noise value is generated in $\BU{1}$ worst-case time using $\BU{\log T}$ words.

\medskip

{\bf Example.}
\Cref{fig: Example of two queries} illustrates an update.
The old active path goes to $b(l)=10010$, and the new active path goes to $b(t)=10111$.
Their longest common prefix is $u=10$, so $m=2$.
The value at level $2$ has not yet been stored; it is sampled by conditioning on the nearest stored prefix sums above and below it, namely $p_0$ and $p_3$.
The final value $p'_5$ is then obtained by adding an independent Gaussian increment for the new branch below $u$.
The resulting state stores exactly the valid prefix sums on the new active path.
Remaining details in the proof of \Cref{thm: Constant look-up time} can be found in \Cref{app:fastgame}.

\begin{figure}[t]
\centering

\begin{subfigure}{0.45\textwidth}
    \centering
    \resizebox{.75\textwidth}{!}{\begin{tikzpicture}[
    roundnode/.style={circle, draw, minimum size=7mm, inner sep=1pt},
]

\node[red!50!black, roundnode] (r) {$p_0$};

\node[below right=1 of r, red!50!black, roundnode] (1) {};
\draw[red!50!black] (r) -- (1);
\node[above=0.00001 of 1, red!50!black, font=\scriptsize] {$1$};

\node[below left=1 of 1, red!50!black, roundnode, text=green!60!black] (10) {$p'_2$};
\draw[red!50!black] (1) -- (10);
\node[above=0.00001 of 10, red!50!black, font=\scriptsize] {$u = 10$};

\node[below left=1 of 10, red!50!black, roundnode] (100) {$p_3$};
\draw[red!50!black] (10) -- (100);
\node[above=0.00001 of 100, red!50!black, font=\scriptsize] {$100$};

\node[below right=1 of 100, red!50!black, roundnode] (1001) {};
\draw[red!50!black] (100) -- (1001);
\node[above=0.00001 of 1001, red!50!black, font=\scriptsize] {$1001$};

\node[below left=1 of 1001, red!50!black, roundnode] (10010) {$p_5$};
\draw[red!50!black] (1001) -- (10010);
\node[below=0.00001 of 10010, red!50!black, font=\scriptsize] {$\text{bin}(l-1) = 10010$};

\node[below right=1 of 10, green!60!black, roundnode] (101) {};
\draw[green!60!black] (10) -- (101); 
\node[above=0.00001 of 101, green!60!black, font=\scriptsize] {$101$};

\node[below right=1 of 101, green!60!black, roundnode] (1011) {};
\draw[green!60!black] (101) -- (1011); 
\node[above=0.00001 of 1011, green!60!black, font=\scriptsize] {$1011$};

\node[below right=1 of 1011, green!60!black, roundnode] (10111) {$p'_5$};
\draw[green!60!black] (1011) -- (10111); 
\node[below=0.00001 of 10111, green!60!black, font=\scriptsize] {$\text{bin}(t-1) = 10111$};

\end{tikzpicture}}
    \caption{The active path at time $l_i$, where entry $i$ was queried last, and the new active path at time $t$. The green part shows where the paths are independent. The longest common prefix of $\text{bin}(l-1)$ and $\text{bin}(t-1)$ is $u=10$. }
\end{subfigure}
\hfill
\begin{subfigure}{0.45\textwidth}
    \centering
    \resizebox{.9\textwidth}{!}{\begin{tikzpicture}

\node at (0,0) {\textcolor{red!50!black}{$P^{(l)}$}};
\foreach \i/\val in {1/p_0,2/,3/,4/p_{3},5/,6/p_5}
{
    \node[draw=red!50!black, minimum width=1cm, minimum height=0.7cm] 
    at (0,-\i*0.8) {\textcolor{red!50!black}{$\val$}};
}

\node at (2,0) {\textcolor{red!50!black}{$w^{(l)}$}};
\foreach \i/\val in {1/1,2/0,3/0,4/1,5/0,6/1}
{
    \node[draw=red!50!black, minimum width=1cm, minimum height=0.7cm] 
    at (2,-\i*0.8) {\textcolor{red!50!black}{$\val$}};
}

\node at (5,0) {\textcolor{green!60!black}{$P^{(t)}$}};
\foreach \i/\val in {1/p_0,2/,3/p'_2,4/p_3,5/,6/p'_5}
{
    \node[draw=green!60!black, minimum width=1cm, minimum height=0.7cm] 
    at (5,-\i*0.8) {\textcolor{green!60!black}{$\val$}};
}

\node at (7,0) {\textcolor{green!60!black}{$w^{(t)}$}};
\foreach \i/\val in {1/1,2/0,3/1,4/0,5/0,6/1}
{
    \node[draw=green!60!black, minimum width=1cm, minimum height=0.7cm] 
    at (7,-\i*0.8) {\textcolor{green!60!black}{$\val$}};
}

\end{tikzpicture}}
    \caption{Left, the structure of $w^{(l)}$ and $P^{(l)}$, which corresponds to the red path in panel (a). Right, the structure of $w^{(t)}$ and $P^{(t)}$ which is the structure of $DS$ after updating from time $l$ to time $t$. The latter corresponds to the green path in panel~(a). The $w$ vector shows which values correspond to prefix sums in the current active path.}
\end{subfigure}

\caption{An example of the structure saved for each index $i \in [k]$ where $\log(T)+1 = 6$.}
\Description{Two panels showing the old and new active paths in a binary tree and the corresponding stored prefix-sum vectors.}
\label{fig: Example of two queries}
\end{figure}
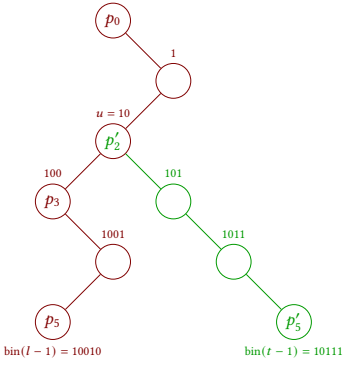
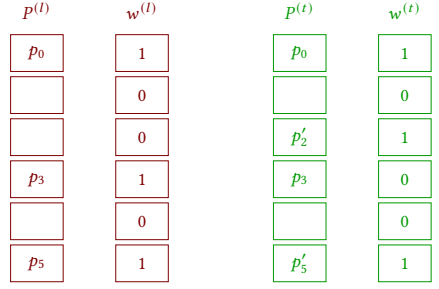

\section{Private Orthogonal Range Queries}
\label{sec:private-ortogonal-range-queries}
Orthogonal range counting is a fundamental problem in databases, with applications in query optimization (for cardinality estimation), OLAP analytics, indexing and spatial search, and approximate query processing.
Its private analogue, private orthogonal range counting, is thus a potentially central primitive in private data analysis and database systems, but private range counting data structures in the literature (with polylogarithmic error) have been rather theoretical and we are not aware of any implementations.

We consider a point multiset in the grid $[B]^d$, where $d$ is considered a constant.
An orthogonal range query is specified by an axis-aligned rectangle $R \subseteq [B]^d$ and asks for the number of points contained in $R$.
Equivalently, the query is defined by $d$ range predicates and returns the number of points satisfying all range predicates.
The problem can be considered in both the static setting where all $n$ points are fixed in advance and the dynamic setting where points can be inserted and deleted.
Changing the viewpoint, we can consider the point set as a vector indexed by positions in $[B]^d$, where each entry contains the number of points in that position.
An update either increments or decrements the vector at a single position.
We do not require the data structure to enforce that deletions correspond to earlier insertions, so the vector may take negative values.
Given a sequence of 1-sparse update vectors $x^{(1)}, x^{(2)}, \dots$, the data structure should answer orthogonal range queries at time $t \leq T$ on the vector $\sum_{i \leq t}x^{(i)}$ while ensuring privacy of point updates.

\subsection{Private Orthogonal Range Queries using Private CountSketches} \label{sec:private-ortogonal-range-queries-using-CS}

{\bf Static data structure.}
We can use a standard dyadic decomposition to construct a collection of Private CountSketches that can be used to answer $d$-dimensional range queries in the static setting.
In each dimension, the interval $[B]$ is decomposed into $\log_2(B)$ dyadic components, each a partition of $[B]$ into intervals of $2^k$ points, for $k \in \{0,\dots,\log_2(B)-1\}$.
Combining these across dimensions results in the $d$-dimensional grid being decomposed into $(\log_2 B)^d$ dyadic components, each a partition into rectangles of size $2^{k_1}\times 2^{k_2}\times\dots\times 2^{k_d}$.

For each component $D$ we consider the vector $y^{(D)}$ that records the number of points in each rectangle.
Each point belongs to exactly one rectangle in each dyadic decomposition, and thus each such vector has sensitivity~1.
For each component we instantiate an independent copy of Private CountSketch (see \Cref{sec:definitions}) with independent noise and hash functions.
We will refer to all these copies as the $\widetilde{CS}$ collection.
Let $n$ denote an upper bound on the number of points so that $\norm{y^{(D)}}_1 \leq n$ for all $D$.
Let $\M$ be the family of $(\log_2 B)^d$ independent instances of $\widetilde{CS}$.

\medskip

{\bf Answering queries.}
It is not hard to show that every axis-aligned rectangle in $[B]^d$ can be represented as a disjoint union of at most $(2\log_2 B)^d$ rectangles from these dyadic components, at most $2^d$ from each.
To answer a query, $\M$ queries the corresponding $\widetilde{CS}$ instances and sums their estimates.
The output is therefore the sum of $(2\log_2 B)^d$ CountSketch estimates.
Each queried $\widetilde{CS}$ entry is an unbiased estimate of the corresponding dyadic count, since the Private CountSketch is unbiased.
Therefore the sum is an unbiased estimate of the range count.

\medskip

{\bf Space guarantees.}
To bound the space usage of Private CountSketch we invoke a fact from the compressive sensing literature (reformulated to fit our notation):
\begin{lemma}[\cite{MathematicalIntroductionToCompressiveSensing}, Theorem 2.5]
    Let $b \in \Z_+$, $y\in \R^d$ where $d\geq b$, and let $\text{tail}_b(y)$ be the vector obtained by removing from $y$ the $b$ largest absolute-value coordinates.
    Then $\norm{\text{tail}_b(y)}_2 \leq \norm{y}_1/(2\sqrt{b})$. \label{Lemma: Lower bound on tail}
\end{lemma}
\vspace{-4mm}
We choose the sketch size to balance the noise arising from privacy and the sketching error for a single instance of $\widetilde{CS}$ as follows:
\begin{lemma}
    Let $y$ be an input vector, let $\sigma^2>0$ be a variance parameter, and let $E\geq \sigma$ be a target error.
    For any $\beta\in(0,1)$, there is a Private CountSketch with noise parameter $\sigma^2$, row size $b=\BU{1+\norm{y}_1/E}$, and $t=\T{\log(1/\beta)}$ repetitions such that every fixed coordinate estimate has error $\BU{E}$ with probability at least $1-\beta$.
    The sketch uses $\BU{t(1+\norm{y}_1/E)}$ words and satisfies $\rho$-zero-concentrated differential privacy with $\rho = t/(2\sigma^2)$.
    \label{Lemma: Space of one instance of PCS}
\end{lemma}
If $E$ is a constant-factor upper bound on $\norm{y}_1$, the sketch can be replaced by the zero estimator, which has error at most $\norm{y}_1$ and uses constant space.
For $b=\Theta\lr{\norm{y}_1/E}$ the CountSketch tail error is at most the target error $E$; in our applications, $E$ is chosen to be comparable to the privacy noise scale.
The proof of \Cref{Lemma: Space of one instance of PCS} can be found in Appendix~\ref{app: lemma space one PCS}.

\medskip

We can now use the fact that the estimate has at most $2^d$ contributions from each of $(2\log_2 B)^d$ independent instances of $\widetilde{CS}$, which means that we can employ a Bernstein inequality to show concentration.
The Gaussian noise must be calibrated to this $\ell_2$ sensitivity, so the full $\widetilde{CS}$ collection satisfies $\rho$-zCDP by the Gaussian mechanism and composition.
More details can be found in \Cref{app: Static M}.

\subsection{Private Orthogonal Range Queries under Continual Observation} \label{sec:range queries under continual observation}

One key property of Private CountSketch~\cite{DifferentiallyPrivateLinearSketches, PCS} is that it is linear and can thus be \emph{updated} when the vector changes.
A naive release of Private CountSketch does not by itself provide scalable privacy under continual observation.
Instead, by using the results in \Cref{sec:FastGaMe}, we will extend the static data structure to the dynamic setting.
We can model the collection of $\widetilde{CS}$ that constitutes the mechanism described in the previous section as a single vector that supports turnstile updates.
By combining the static structure and the mechanism FastGaMe described in \Cref{sec:FastGaMe}, we get a data structure for answering $d$-dimensional range queries under continual observation with the following guarantees:
\begin{theorem}
      For any choice of positive integers $d \leq B$, $n$, and for any $\beta \in (0,1/\log B), \; \rho > 0$, let
      $E = \BL{d(\log B)^{d+1}\log(n)\log(d/\beta)/\sqrt{\rho}}$
      denote some fixed error.
      Let $\M_{CC}$ be the dynamic range counting mechanism supporting $n$ updates,
      combining $(\log_2 B)^d$ instances of Private CountSketch with the FastGaMe mechanism.
     Then $\M_{CC}$ can be represented as a data structure with the following properties:
    \begin{enumerate}
            \item[$\circ$] It answers $d$-dimensional orthogonal range queries over $[B]^d$ under continual observation with worst-case additive estimation error $\BU{E}$ with probability at least $1 - \beta$.
            \item[$\circ$] The answers to any sequence of queries satisfy $\rho$-zCDP with respect to point updates.
            \item[$\circ$] Its space usage is $\BU{\frac{n(\log B)^d\log(d/\beta)\log(n)}{E} + 1}$ words.
            \item[$\circ$] Each query and update takes $\BU{(\log B)^d\max\lr{\log(d/\beta), d\log(B)}}$ worst-case time.
            \SiSe{Check this again + that it persists}
    \end{enumerate}
    \label{Theorem: zCDP upper bound under continual observation}
\end{theorem}
\Cref{Theorem: zCDP upper bound under continual observation} follows from the observation that the $\widetilde{CS}$ collection can be treated as a single large vector of counters with $\ell_2$ sensitivity $\sqrt{(\log_2 B)^d t}$, where $t$ is the number of repetitions in Private CountSketch.
Hence, we can combine the static mechanism from \Cref{sec:private-ortogonal-range-queries-using-CS} with FastGaMe from \Cref{sec:FastGaMe} to allow continual release over $T=n$ steps.
To maintain the same privacy parameter, the parameter $\sigma^2$ used should be increased by a factor $\BU{\log n}$ relative to the static case.
Space increases by a multiplicative factor $\BU{\log(n)}$ due to the space requirements of FastGaMe, storing $\BU{\log(n)}$ noise values for each entry in the sketch.
At each step, the output distribution is the same as in the static setting with these adjusted space and noise parameters.

\section{Private Join Size Estimation}
\label{sec:private-join-size-estimation}
\providecommand{\Var}{\mathrm{Var}}
\providecommand{\Prb}{\mathrm{Pr}}
\providecommand{\CS}{\mathrm{CS}}
\providecommand{\med}{\mathrm{median}}
\providecommand{\ip}[2]{\langle #1,#2 \rangle}

Join size estimation is a central primitive in streaming algorithms and database systems.
Given two relations $R$ and $S$ over a common key domain, let $f_R(a)$ and $f_S(a)$ denote the frequency of key $a$ in the two relations and let $x,y \in \R^d$ be the frequency vectors. 
The join size is
\[
J(R,S) = \sum_{a} f_R(a)f_S(a) = \ip{x}{y} \enspace .
\]
The case where $R$ and $S$ are identical is called a \emph{self-join}, and the problem then becomes that of estimating the second frequency moment ($F_2$).

To estimate the inner product $\ip{x}{y}$ from the CountSketches $\CS(x)$ and $\CS(y)$ (defined in~\Cref{sec:definitions}) we use the same hash functions for both sketches~\cite{AGMS99,WDLAS09}.
The standard one-row inner-product estimator is
\(
Z_i = \ip{\CS_i(x)}{\CS_i(y)}
\), and the overall estimator is the median:
\(
\widehat J = \med\{Z_1,\dots,Z_k\}
\).

Since CountSketches can be maintained efficiently under updates, using them is a natural approach for approximating join sizes in dynamic settings.
Recent work has shown guarantees for CountSketch estimators based on taking the median of a small number of rows $k$~\citep{LPT21}.
\begin{theorem}[Larsen--Pagh--Tet{\v e}k~\cite{LPT21}]\label{thm:join-countsketch-median}
Let \(Z=\ip{\CS_1(x)}{\CS_1(y)}\) denote a single-row Count\-Sketch inner-product estimator using fully random hashing, so that
\(\E[Z]=\ip{x}{y}\), and let $\widehat J$ denote the median of \(k\) independent copies of \(Z\), where $k \geq 3$ is odd.
For \(k=3\)
\[
\Var(\widehat J)
\le
\min\left\{
\BU{\frac{\norm{x}_1^2\norm{y}_1^2}{b^2}},
\BU{\frac{\norm{x}_2^2\norm{y}_2^2}{b}}
\right\}.
\]
More generally, for every odd $k \ge 3$ and integer $q \ge 1$:
\[
\E\left[|\widehat J-\ip{x}{y}|^{(k+1)q/2}\right]
\le
\tfrac{2^k}{\sqrt{k}} \cdot
\E\left[|Z-\ip{x}{y}|^q\right]^{(k+1)/2} \enspace .
\]
\end{theorem}

Thus the median operation effectively raises the moment decay to the power $(k+1)/2$.
Combined with one-row moment bounds this yields tail bounds as a function of the sketch width $b$.

\subsection{Join Size Estimation Using Private CountSketch}

\citet{PCS} gave a precise analysis of the error distribution of \emph{single-entry estimates}~$\tilde x_a$ obtained from a Private CountSketch with rows
\(\widetilde{\CS}_i(x)=\CS_i(x)+\nu_i\), where
\(
\nu_i \sim N(0,\sigma^2 I_b)
\), 
but did not extend this analysis to inner products.
We extend the analysis of Private CountSketch to inner product estimates.
Again, it is important to use the same hash functions for the sketches of $x$ and $y$, but we add independent Gaussian noise to all sketch entries, ensuring privacy of both relations.
That is, for each row $i$, let
\(\widetilde{\CS}_i(y)=\CS_i(y)+\xi_i\)
where
\(
\xi_i \sim N(0,\sigma^2 I_b)
\) is chosen independently of the noise vector \(\nu\) used for \(\CS(x)\).
The one-row inner-product estimator is
\(
\widetilde Z_i
=
\ip{\widetilde{\CS}_i(x)}{\widetilde{\CS}_i(y)}
\) and the corresponding median estimator is
\(
\widetilde J = \med\{\widetilde Z_1,\dots,\widetilde Z_k\}
\).
We show the following bound on the accuracy of $\widetilde J$:

\begin{lemma}\label{lem:private-median-regimes}
Let $k\ge 3$ be odd, let $\eta\in(0,1/2)$, and set
\(
C_\eta=\eta^{-2/(k+1)}
\).
Assume that $b\ge 8C_\eta$. Then
\(
\Prb\left[|\widetilde J-\ip{x}{y}|> c\,\Delta_{x,y,\eta,k,b}\right]\le \eta
\)
where $c$ is a universal constant and
\[
\Delta_{x,y,\eta,k,b}
=
\min\left\{
    \sqrt{\frac{C_\eta}{b}}\,\norm{x}_2\norm{y}_2,
    \frac{C_\eta}{b}\,\norm{x}_1\norm{y}_1
\right\}
+ \sigma(\norm{x}_2+\norm{y}_2)\sqrt{\ln(1/\eta)}
+ \sigma^2\sqrt{b\ln(1/\eta)}
\enspace .
\]
\end{lemma}

The three terms in the error bound $\Delta_{x,y,\eta,k,b}$ are, respectively, the median-amplified Count\-Sketch error, the sketch--noise cross term, and the noise--noise product term.
The CountSketch term switches from an $\ell_2$ bound to an $\ell_1$ bound when $b \approx C_\eta\norm{x}_1^2\norm{y}_1^2/(\norm{x}_2^2\norm{y}_2^2)$.
After this point the CountSketch contribution decreases as $\BU{1/b}$, while the noise--noise term increases as $\BU{\sqrt{b}}$; these balance at $b \approx (C_\eta\norm{x}_1\norm{y}_1/(\sigma^2\sqrt{\log(1/\eta)}))^{2/3}$.
Depending on the skew of $x$ and $y$ and the noise scale~$\sigma$, some parameter regimes may be empty, and the sketch--noise term may dominate as a baseline privacy error.

To show \Cref{lem:private-median-regimes} we need a bound on the single-row accuracy of (non-private) CountSketch:
\begin{lemma}\label{lem:one-row-cs-concentration}
Fix $x,y\in\R^d$ and let $\eta\in(0,1/2)$. For each pair of CountSketch rows $\CS_i(x)$ and $\CS_i(y)$, the following bounds simultaneously hold with probability at least $1-\eta$:
\begin{equation}\label{eq:one-row-cs-concentration}
\begin{aligned}
\left|\ip{\CS_i(x)}{\CS_i(y)}-\ip{x}{y}\right|
&\le
\min\left\{
\sqrt{\tfrac{8}{b\eta}}\norm{x}_2\norm{y}_2,
\tfrac{4}{b\eta}\norm{x}_1\norm{y}_1
\right\}
\\
\norm{\CS_i(x)}_2^2
\le
\left(1+\sqrt{\tfrac{8}{b\eta}}\right)\norm{x}_2^2
\qquad & \qquad
\norm{\CS_i(y)}_2^2
\le
\left(1+\sqrt{\tfrac{8}{b\eta}}\right)\norm{y}_2^2
\end{aligned}
\enspace .
\end{equation}
\end{lemma}

\begin{proof}
The lemma follows along the lines of the CountSketch analysis in~\cite{CountSketchCCF,LPT21}.
For every $x,y\in\R^d$ we have
\(\E[\ip{\CS_i(x)}{\CS_i(y)}]=\ip{x}{y}\) and
\(\Var(\ip{\CS_i(x)}{\CS_i(y)})\le \tfrac{2}{b}\norm{x}_2^2\norm{y}_2^2\).
Chebyshev's inequality with failure probability $\eta/4$ gives
\[
\Prb\left[
\left|\ip{\CS_i(x)}{\CS_i(y)}-\ip{x}{y}\right|>
\sqrt{\tfrac{8}{b\eta}}\norm{x}_2\norm{y}_2
\right]
\le \eta/4 .
\]
For the $L_1$ bound we expand \(\E\left[\left|\ip{\CS_i(x)}{\CS_i(y)}-\ip{x}{y}\right|\right]\):
\[
\begin{aligned}
\E\left[
\sum_{p\neq q}
\mathbf 1\{h_i(p)=h_i(q)\}s_i(p)s_i(q)x_p y_q
\right]
&\le
\frac{1}{b}\sum_{p\neq q}|x_p||y_q|
\le
\norm{x}_1\norm{y}_1 / b \enspace .
\end{aligned}
\]
Markov's inequality gives
\(
\Prb\left[
\left|\ip{\CS_i(x)}{\CS_i(y)}-\ip{x}{y}\right|>
\tfrac{4}{b\eta}\norm{x}_1\norm{y}_1
\right]
\le \eta/4
\).
Applying the same variance bound and Chebyshev's inequality to $(x,x)$ gives the stated bound on
\(\norm{\CS_i(x)}_2^2\), with failure probability at most $\eta/4$, and similarly for $y$.
A union bound over the four events proves the claim.
\end{proof}

We now turn to the estimators $\widetilde Z_i$ based on the rows of Private CountSketch.
\begin{lemma}\label{lem:conditional-concentration}
For $\eta\in(0,1/2)$, $b\ge 1/\eta$, and $x,y\in\R^d$, each row estimator $\widetilde Z_i$ satisfies:
\[
\Prb\left[
|\widetilde Z_i-\ip{x}{y}|>
\min\left\{
\sqrt{\tfrac{8}{b\eta}}\norm{x}_2\norm{y}_2,
\tfrac{4}{b\eta}\norm{x}_1\norm{y}_1
\right\}
+ 
\Lambda^{\mathrm{priv}}_{\eta}
\right]
\le 2\eta,
\]
where
\(
\Lambda^{\mathrm{priv}}_{\eta}
=
3\sigma(\norm{x}_2+\norm{y}_2 + 5\sigma\sqrt{b})\sqrt{\ln(4/\eta)}
\).
\end{lemma}

The lemma says that a single row fails with probability at most $2\eta$ when the error threshold is chosen at the corresponding scale. 
The first term is the CountSketch error from~\citep{LPT21}. 
The terms in $\Lambda^{\mathrm{priv}}_{\eta}$ are privacy-related error terms.

\begin{proof}
By the triangle inequality,
\begin{equation}\label{eq:triangle-inequality}
|\widetilde Z_i-\ip{x}{y}|
\le
\left|\ip{\CS_i(x)}{\CS_i(y)}-\ip{x}{y}\right|
+ |\ip{\CS_i(x)}{\xi_i} + \ip{\CS_i(y)}{\nu_i}| 
+ |\ip{\nu_i}{\xi_i}| \enspace .
\end{equation}
By Lemma~\ref{lem:one-row-cs-concentration}, (\ref{eq:one-row-cs-concentration}) holds with probability at least $1-\eta$, and bounds the first term of~(\ref{eq:triangle-inequality}).
On that event (fixing all hash functions), the first term is bounded, so it suffices to show
\[
\Prb\left[
|\ip{\CS_i(x)}{\xi_i} + \ip{\CS_i(y)}{\nu_i}| + |\ip{\nu_i}{\xi_i}| >
\Lambda^{\mathrm{priv}}_{\eta}
\;\middle|\;
\text{ event (\ref{eq:one-row-cs-concentration}) holds}\right]
\le \eta,
\]
from which a union bound gives the desired probability upper bound $2\eta$.
Since $b\ge 1/\eta$, the $L_2$ row-norm bounds in~(\ref{eq:one-row-cs-concentration}) imply \(\norm{\CS_i(x)}_2\le 2\norm{x}_2\) and \(\norm{\CS_i(y)}_2\le 2\norm{y}_2\).
Thus the sum \(\ip{\CS_i(x)}{\xi_i}+\ip{\CS_i(y)}{\nu_i}\) is Gaussian with variance \(\sigma^2\norm{\CS_i(x)}_2^2 + \sigma^2\norm{\CS_i(y)}_2^2 \leq 4\sigma^2\norm{x}_2^2 + 4\sigma^2\norm{y}_2^2\). Thus a Gaussian tail bound gives
\(
\Prb\left[
\left|\ip{\CS_i(x)}{\xi_i}+\ip{\CS_i(y)}{\nu_i}\right| >
3\sigma(\norm{x}_2+\norm{y}_2)\sqrt{\ln(4/\eta)}
\right]
\le \eta/2
\).
Finally, \(\ip{\nu_i}{\xi_i}\)
is a sum of products of independent $N(0,\sigma^2)$ variables.
After scaling by \(\sigma^2\), this is the same random variable that arises in the Gaussian Johnson--Lindenstrauss inner-product estimate for two orthogonal unit vectors, see, e.g.,~\cite[Lemma~9.8 and Theorem~9.9]{MathematicalIntroductionToCompressiveSensing}.
Each product term is sub-exponential, and Bernstein's inequality together with
\(\ln(4/\eta)\le 2b\) gives
\[
\Prb\left[
|\ip{\nu_i}{\xi_i}|>
15\sigma^2\sqrt{b\ln(4/\eta)}
\right]
\le \eta/2
 \enspace . \qedhere \]
\end{proof}

\begin{lemma}\label{lem:median-amplification}
Let $k\geq 3$ be odd, let \(\eta\in(0,1)\), and suppose
\(
b\ge 8\eta^{-2/(k+1)}
\).
Then
\[
\Prb\Bigg[
|\med_i \widetilde Z_i-\ip{x}{y}|
>
\Delta_{x,y,\eta,k,b}
\Bigg]
\le \eta
\]
where \(\Delta_{x,y,\eta,k,b} = \min\left\{ \sqrt{\tfrac{64}{b\,\eta^{2/(k+1)}}}\norm{x}_2\norm{y}_2, \tfrac{32\,\norm{x}_1\norm{y}_1}{b\,\eta^{2/(k+1)}} \right\}
+
3\sigma(\norm{x}_2+\norm{y}_2+5\sigma\sqrt b)
\sqrt{\ln\left(32\eta^{-2/(k+1)}\right)}\).
\end{lemma}

\begin{proof}
Set
\(
\alpha=
\tfrac12\,2^{-2k/(k+1)}\eta^{2/(k+1)}
\).
Then \(\alpha<1/2\) and the assumption on \(b\) implies \(b\ge 1/\alpha\).
Let
\(
T_\alpha=
\min\left\{ \sqrt{\tfrac{8}{b\alpha}}\norm{x}_2\norm{y}_2, \tfrac{4}{b\alpha}\norm{x}_1\norm{y}_1 \right\}
+
\Lambda^{\mathrm{priv}}_{\alpha}
\)
and
\(
B_i=\mathbf 1\{|\widetilde Z_i-\ip{x}{y}|>T_\alpha\}
\).
The variables $B_i$ are independent and Lemma~\ref{lem:conditional-concentration} states that $\Prb[B_i=1]\le 2\alpha$. 
Since $k$ is odd, the median fails only if at least $(k+1)/2$ rows fail. This gives
\[
\Prb\left[
|\med_i \widetilde Z_i-\ip{x}{y}|>T_\alpha
\right]
\le
\sum_{j\ge (k+1)/2}\binom{k}{j}(2\alpha)^j
\le
2^k(2\alpha)^{(k+1)/2}
=\eta \enspace . \qedhere
\]
\end{proof}

\subsection{Maintaining a private join estimate under updates}

We now describe how to maintain the noisy join estimate under continual observation.

\begin{theorem}\label{thm:continual-private-join}
Let $T$ be the time horizon, let $b\geq 1$, let $k\geq 3$ be odd, and let $\rho>0$.
Define
\(
    \sigma_{\mathrm{co}}^2 = k(\log_2(T)+2)^2/(8\rho)
\).
Consider a dynamic stream of updates to two relations $R$ and $S$, where each update changes the frequency of a key in one relation by at most $1$ (increase or decrease).
There is a data structure that maintains estimates $\widetilde J_t$ of the join size $J_t=\ip{x_t}{y_t}$ at each time step $t\in[T]$ with the following guarantees:
\begin{itemize}
\item The worst-case update time is $\BU{k}$ in a Word RAM model (see \Cref{app: Preliminaries}).
\item The space usage is $\BU{kb\log T}$ words, in addition to storing the CountSketch hash functions.
\item The entire sequence of released estimates is $\rho$-zCDP under continual observation.
\item For every fixed time $t$ and every $\eta\in(0,1/2)$ such that $b\geq 8\eta^{-2/(k+1)}$,
letting 
\[
\begin{aligned}
\Delta_t(\eta) = \min\left\{
    \sqrt{\tfrac{\eta^{-2/(k+1)}}{b}}\,\norm{x_t}_2\norm{y_t}_2,
    \tfrac{\eta^{-2/(k+1)\norm{x_t}_1\norm{y_t}_1}}{b}
\right\}
+ (\sigma_{\mathrm{co}}(\norm{x_t}_2+\norm{y_t}_2)
+ \sigma_{\mathrm{co}}^2\sqrt{b}) \sqrt{\ln(1/\eta)},
\end{aligned}
\]
we have
\(
\Prb\left[
|\widetilde J_t-\ip{x_t}{y_t}|
>
c\,\Delta_t(\eta)
\right]
\leq \eta ,
\)
where $c$ is a universal constant.
\end{itemize}
\end{theorem}

Note that if estimates are released at all times $t\in[T]$ and $b\geq 8(T/\eta)^{2/(k+1)}$, then with probability at least $1-\eta$ the bound $|\widetilde J_t-\ip{x_t}{y_t}|\leq c\,\Delta_t(\eta/T)$ holds simultaneously for all $t$.
Unlike the point estimates in~\citep{LPT21}, we note that the error grows with the number of repetitions $k$, since $\sigma_{\mathrm{co}}^2$ is proportional to $k$.
The condition $b\geq 8\eta^{-2/(k+1)}$ is satisfied for $k \gtrsim \log(1/\eta)/(2\log b)$, so error probability polynomially small in $b$ can be achieved with constant $k$.

\begin{proof}
We consider all rows of both sketches as one large vector and apply FastGaMe (\Cref{sec:FastGaMe}) to it.
The space usage for FastGaMe is $\BU{\log T}$ per sketch entry.
An update to one relation changes one counter in each of $k$ CountSketch rows, so the combined $\ell_2$ sensitivity of the vectors of all sketch counters is at most $\sqrt{k}$.
Thus, setting the internal node variance in FastGaMe to
\(
    \sigma_{\mathrm{node}}^2
    =
    k(\log_2(T)+2)/(8\rho),
\)
yields $\rho$-zCDP of the sketches for both relations under continual observation.
Since the estimates are computed by post-processing the same guarantee applies to them.

For an update to key $a$, only one bucket in each CountSketch row can change.
For each of the $k$ rows we refresh the necessary noisy sketch entries using FastGaMe and update the maintained row inner product by replacing the affected bucket contribution.
Each refresh and arithmetic update takes $\BU{1}$ worst-case time, giving $\BU{k}$ update time, including recomputing the median.
Notice that this causes the noise of sketch entries to have different timestamps according to when they were last updated.
However, this does not matter: At any fixed time $t$, the noisy sketch entries used in the maintained estimator are independent Gaussians with variance $\sigma_{\mathrm{node}}^2 (\log(T)+1) < \sigma_{\mathrm{co}}^2$.
Thus \Cref{lem:private-median-regimes} applies with $x=x_t$, $y=y_t$, and $\sigma=\sigma_{\mathrm{co}}$.
The simultaneous guarantee follows by a union bound over all $T$ released estimates, which holds even though noise values are correlated over time.
\end{proof}

\section{Conclusion and Open Questions}
We have shown that dynamic data structures with bounded-norm updates, for example linear sketches, support continual observation with no asymptotic time overhead.
Values are released with correlated Gaussian noise of magnitude within a small constant factor of the best continual observation techniques.
It would be interesting to see if similar results hold for more sophisticated low-space matrix factorization techniques with improved error, such as the Buffered Linear Toeplitz matrices of~\citet{Dvijotham2024efficient} or (perhaps more directly) for the higher-degree tree-based methods like that of~\citet{ImprovedCountingUnderContinualObservationWithPureDifferentialPrivacy}.
Though we have focused on methods with Gaussian noise, it appears likely that analogous results hold under Laplace noise, with pure differential privacy guarantees.

Our two applications show that while continual observation comes at a cost in error compared to a one-shot release, the cost in terms of processing time is constant.
They also illustrate the power and versatility of Private CountSketch in data management applications such as multidimensional indexing and sketch-based join size estimation. There are many other settings in which approximate representation of sparse vectors is a key ingredient, so there are likely more such applications.

\begin{acks}
The authors were supported by a Data Science Distinguished Investigator grant from Novo Nordisk Fonden, and are part of BARC, supported by the VILLUM Foundation grant 54451.
\end{acks}

\newpage
\bibliography{bib2doi}

@inproceedings{PCS,
  author = {Pagh, Rasmus and Thorup, Mikkel},
  booktitle = {Advances in Neural Information Processing Systems (NeurIPS)},
  pages = {25631--25643},
  publisher = {Curran Associates, Inc.},
  title = {Improved Utility Analysis of Private CountSketch},
  url = {https://proceedings.neurips.cc/paper_files/paper/2022/file/a47f5cdff1469751597d78e803fc590f-Paper-Conference.pdf},
  volume = {35},
  year = {2022},
  issn = {1049-5258},
  timestamp = {Mon, 08 Jan 2024 00:00:00 +0100},
  biburl = {https://dblp.org/rec/conf/nips/PaghT22.bib},
  bibsource = {dblp computer science bibliography, https://dblp.org},
  _bib2doi_selected = {dblp:/rec/conf/nips/PaghT22.bib},
  _bib2doi_confirmed = {true},
  _bib2doi_finished = {true},
}

@inproceedings{PureDPRecQueries,
  author = {Dwork, Cynthia and Naor, Moni and Reingold, Omer and Rothblum, Guy N.},
  title = {Pure Differential Privacy for Rectangle Queries via Private Partitions},
  booktitle = {Advances in Cryptology -- ASIACRYPT 2015},
  year = {2015},
  publisher = {Springer Berlin Heidelberg},
  address = {Berlin, Heidelberg},
  pages = {735--751},
  isbn = {978-3-662-48800-3},
  doi = {10.1007/978-3-662-48800-3\_30},
  timestamp = {Thu, 27 Jul 2023 01:00:00 +0200},
  biburl = {https://dblp.org/rec/conf/asiacrypt/DworkNRR15.bib},
  bibsource = {dblp computer science bibliography, https://dblp.org},
  volume = {9453},
  url = {https://doi.org/10.1007/978-3-662-48800-3\_30},
  _bib2doi_selected = {dblp:/rec/conf/asiacrypt/DworkNRR15.bib},
  _bib2doi_confirmed = {true},
  _bib2doi_finished = {true},
}

@article{CountSketchCCF,
  title = {Finding frequent items in data streams},
  journal = {Theoretical Computer Science},
  volume = {312},
  number = {1},
  pages = {3-15},
  year = {2004},
  note = {Conference version in ICALP 2002.},
  issn = {0304-3975},
  doi = {10.1016/S0304-3975(03)00400-6},
  url = {https://doi.org/10.1016/S0304-3975(03)00400-6},
  author = {Moses Charikar and Kevin Chen and Martin Farach-Colton},
  timestamp = {Sat, 30 Sep 2023 01:00:00 +0200},
  biburl = {https://dblp.org/rec/journals/tcs/CharikarCF04.bib},
  bibsource = {dblp computer science bibliography, https://dblp.org},
  _bib2doi_old_doi = {https://doi.org/10.1016/S0304-3975(03)00400-6},
  _bib2doi_selected = {dblp:/rec/journals/tcs/CharikarCF04.bib},
  _bib2doi_confirmed = {true},
  _bib2doi_finished = {true},
}

@inproceedings{ImprovedConcentrationBoundsCountSketch,
  author = {Gregory T. Minton and Eric Price},
  title = {Improved Concentration Bounds for Count-Sketch},
  booktitle = {Proceedings of the 2014 Annual ACM-SIAM Symposium on Discrete Algorithms (SODA)},
  pages = {669-686},
  doi = {10.1137/1.9781611973402.51},
  url = {https://epubs.siam.org/doi/abs/10.1137/1.9781611973402.51},
  year = {2014},
  timestamp = {Mon, 25 Apr 2022 01:00:00 +0200},
  biburl = {https://dblp.org/rec/conf/soda/MintonP14.bib},
  bibsource = {dblp computer science bibliography, https://dblp.org},
  _bib2doi_selected = {dblp:/rec/conf/soda/MintonP14.bib},
  _bib2doi_confirmed = {true},
}

@inproceedings{DifferentiallyPrivateLinearSketches,
  author = {Fuheng Zhao and Dan Qiao and Rachel Redberg and Divyakant Agrawal and Amr El Abbadi and Yu{-}Xiang Wang},
  title = {Differentially Private Linear Sketches: Efficient Implementations and Applications},
  booktitle = {Advances in Neural Information Processing Systems (NeurIPS)},
  year = {2022},
  url = {http://papers.nips.cc/paper_files/paper/2022/hash/525338e0d98401a62950bc7c454eb83d-Abstract-Conference.html},
  timestamp = {Mon, 08 Jan 2024 00:00:00 +0100},
  biburl = {https://dblp.org/rec/conf/nips/ZhaoQRAAW22.bib},
  bibsource = {dblp computer science bibliography, https://dblp.org},
  _bib2doi_selected = {dblp:/rec/conf/nips/ZhaoQRAAW22.bib},
  _bib2doi_confirmed = {true},
}

@book{MathematicalIntroductionToCompressiveSensing,
  author = {Simon Foucart and Holger Rauhut},
  title = {A Mathematical Introduction to Compressive Sensing},
  series = {Applied and Numerical Harmonic Analysis},
  publisher = {Birkh{\"{a}}user},
  year = {2013},
  url = {https://doi.org/10.1007/978-0-8176-4948-7},
  doi = {10.1007/978-0-8176-4948-7},
  isbn = {978-0-8176-4947-0},
  timestamp = {Wed, 14 Jun 2017 01:00:00 +0200},
  biburl = {https://dblp.org/rec/books/daglib/0036092.bib},
  bibsource = {dblp computer science bibliography, https://dblp.org},
  _bib2doi_selected = {dblp:/rec/books/daglib/0036092.bib},
  _bib2doi_confirmed = {true},
}

@inproceedings{Bun2026ConcentratedDP,
  author = {Bun, Mark and Steinke, Thomas},
  editor = {Hirt, Martin and Smith, Adam},
  title = {Concentrated Differential Privacy: Simplifications, Extensions, and Lower Bounds},
  booktitle = {Theory of Cryptography},
  year = {2016},
  publisher = {Springer Berlin Heidelberg},
  address = {Berlin, Heidelberg},
  pages = {635--658},
  isbn = {978-3-662-53641-4},
  timestamp = {Sun, 14 Mar 2021 00:00:00 +0100},
  biburl = {https://dblp.org/rec/conf/tcc/BunS16.bib},
  bibsource = {dblp computer science bibliography, https://dblp.org},
  doi = {10.1007/978-3-662-53641-4_24},
  _bib2doi_selected = {dblp:/rec/conf/tcc/BunS16.bib},
  _bib2doi_confirmed = {true},
}

@inproceedings{OptimalPrivateHalfspaceCountingViaDiscrepancy,
  title = {Optimal private halfspace counting via discrepancy},
  author = {Muthukrishnan, Shanmugavelayutham and Nikolov, Aleksandar},
  booktitle = {Proceedings of the forty-fourth annual ACM symposium on Theory of computing},
  pages = {1285--1292},
  year = {2012},
  timestamp = {Mon, 03 Mar 2025 00:00:00 +0100},
  biburl = {https://dblp.org/rec/conf/stoc/MuthukrishnanN12.bib},
  bibsource = {dblp computer science bibliography, https://dblp.org},
  doi = {10.1145/2213977.2214090},
  url = {https://doi.org/10.1145/2213977.2214090},
  _bib2doi_selected = {dblp:/rec/conf/stoc/MuthukrishnanN12.bib},
  _bib2doi_confirmed = {true},
  _bib2doi_finished = {true},
}

@article{Chan2011ReleaseStatistics,
  author = {Chan, T.-H. Hubert and Shi, Elaine and Song, Dawn},
  title = {Private and Continual Release of Statistics},
  year = {2011},
  issue_date = {November 2011},
  publisher = {Association for Computing Machinery},
  address = {New York, NY, USA},
  volume = {14},
  number = {3},
  issn = {1094-9224},
  url = {https://doi.org/10.1145/2043621.2043626},
  doi = {10.1145/2043621.2043626},
  abstract = {We ask the question: how can Web sites and data aggregators continually release updated statistics, and meanwhile preserve each individual user’s privacy? Suppose we are given a stream of 0’s and 1’s. We propose a differentially private continual counter that outputs at every time step the approximate number of 1’s seen thus far. Our counter construction has error that is only poly-log in the number of time steps. We can extend the basic counter construction to allow Web sites to continually give top-k and hot items suggestions while preserving users’ privacy.},
  journal = {ACM Trans. Inf. Syst. Secur.},
  month = {nov},
  articleno = {26},
  numpages = {24},
  keywords = {streaming algorithm, continual mechanism, Differential privacy},
  timestamp = {Tue, 06 Nov 2018 00:00:00 +0100},
  biburl = {https://dblp.org/rec/journals/tissec/ChanSS11.bib},
  bibsource = {dblp computer science bibliography, https://dblp.org},
  _bib2doi_selected = {dblp:/rec/journals/tissec/ChanSS11.bib},
  _bib2doi_confirmed = {true},
  note = {Appeared in Cryptology {ePrint} Archive 2010/076 and ICALP 2010.},
}

@article{TheAlgorithmicFoundationsofDifferentialPrivacy,
  author = {Dwork, Cynthia and Roth, Aaron},
  title = {The Algorithmic Foundations of Differential Privacy},
  year = {2014},
  issue_date = {Aug 2014},
  publisher = {Now Publishers Inc.},
  address = {Hanover, MA, USA},
  volume = {9},
  number = {3-4},
  issn = {1551-305X},
  url = {https://doi.org/10.1561/0400000042},
  doi = {10.1561/0400000042},
  journal = {Found. Trends Theor. Comput. Sci.},
  month = {aug},
  pages = {211--407},
  numpages = {197},
  timestamp = {Sat, 30 Sep 2023 01:00:00 +0200},
  biburl = {https://dblp.org/rec/journals/fttcs/DworkR14.bib},
  bibsource = {dblp computer science bibliography, https://dblp.org},
  _bib2doi_selected = {dblp:/rec/journals/fttcs/DworkR14.bib},
  _bib2doi_confirmed = {true},
  _bib2doi_finished = {true},
}

@article{Xiao2011Wavelet,
  author = {Xiaokui Xiao and Guozhang Wang and Johannes Gehrke},
  title = {Differential Privacy via Wavelet Transforms},
  journal = {IEEE Transactions on Knowledge and Data Engineering},
  volume = {23},
  number = {8},
  pages = {1200--1214},
  year = {2011},
  doi = {10.1109/TKDE.2010.247},
  url = {https://doi.org/10.1109/TKDE.2010.247},
  bibsource = {dblp computer science bibliography, https://dblp.org},
  note = {Appeared as arXiv 0909.5530 and in ICDE 2010.},
  timestamp = {Thu, 08 Jun 2017 01:00:00 +0200},
  biburl = {https://dblp.org/rec/journals/tkde/XiaoWG11.bib},
  _bib2doi_selected = {dblp:/rec/journals/tkde/XiaoWG11.bib},
  _bib2doi_confirmed = {true},
}

@article{FactorizationNormsAndHereditaryDiscrepancy,
  author = {Ji{\v{r}}{\i} Matou{\v{s}}ek and Aleksandar Nikolov and Kunal Talwar},
  title = {Factorization Norms and Hereditary Discrepancy},
  journal = {International Mathematics Research Notices},
  volume = {2020},
  number = {3},
  pages = {751--780},
  year = {2020},
  doi = {10.1093/imrn/rny033},
  url = {https://doi.org/10.1093/imrn/rny033},
  bibsource = {dblp computer science bibliography, https://dblp.org},
  timestamp = {Thu, 24 Jan 2019 00:00:00 +0100},
  biburl = {https://dblp.org/rec/journals/corr/MatousekNT14.bib},
  _bib2doi_selected = {dblp:/rec/journals/corr/MatousekNT14.bib},
  _bib2doi_confirmed = {true},
  _bib2doi_finished = {true},
}

@article{Dwork2010ContinalObs,
  author = {Dwork, Cynthia and Naor, Moni and Pitassi, Toniann and Rothblum, Guy N.},
  title = {Differential Privacy Under Continual Observation},
  journal = {STOC},
  year = {2010},
  timestamp = {Wed, 14 Nov 2018 00:00:00 +0100},
  biburl = {https://dblp.org/rec/conf/stoc/DworkNPR10.bib},
  bibsource = {dblp computer science bibliography, https://dblp.org},
  doi = {10.1145/1806689.1806787},
  url = {https://doi.org/10.1145/1806689.1806787},
  _bib2doi_selected = {dblp:/rec/conf/stoc/DworkNPR10.bib},
  _bib2doi_confirmed = {true},
  _bib2doi_finished = {true},
}

@inproceedings{DifferentialPrivacyonFullyDynamicStreams,
  author = {Qiu, Yuan and Yi, Ke},
  title = {Differential Privacy on Fully Dynamic Streams},
  booktitle = {Advances in Neural Information Processing Systems (NeurIPS)},
  year = {2025},
  _bib2doi_finished = {true},
}

@inproceedings{AGMS99,
  author = {Alon, Noga and Gibbons, Phillip B. and Matias, Yossi and Szegedy, Mario},
  title = {Tracking Join and Self-Join Sizes in Limited Storage},
  booktitle = {Proceedings of the Eighteenth ACM SIGMOD-SIGACT-SIGART Symposium on Principles of Database Systems (PODS)},
  pages = {10--20},
  year = {1999},
  timestamp = {Thu, 19 Feb 2026 00:00:00 +0100},
  biburl = {https://dblp.org/rec/conf/pods/AlonGMS99.bib},
  bibsource = {dblp computer science bibliography, https://dblp.org},
  doi = {10.1145/303976.303978},
  _bib2doi_selected = {dblp:/rec/conf/pods/AlonGMS99.bib},
  _bib2doi_confirmed = {true},
}

@inproceedings{LPT21,
  author = {Larsen, Kasper Green and Pagh, Rasmus and Tet{\v e}k, Jakub},
  title = {CountSketches, Feature Hashing and the Median of Three},
  booktitle = {Proceedings of the 38th International Conference on Machine Learning (ICML)},
  pages = {6011--6020},
  year = {2021},
  timestamp = {Wed, 25 Aug 2021 01:00:00 +0200},
  biburl = {https://dblp.org/rec/conf/icml/LarsenPT21.bib},
  bibsource = {dblp computer science bibliography, https://dblp.org},
  url = {http://proceedings.mlr.press/v139/larsen21a.html},
  _bib2doi_selected = {dblp:/rec/conf/icml/LarsenPT21.bib},
  _bib2doi_confirmed = {true},
  _bib2doi_finished = {true},
}

@inproceedings{WDLAS09,
  author = {Weinberger, Kilian and Dasgupta, Anirban and Langford, John and Smola, Alex and Attenberg, Josh},
  title = {Feature Hashing for Large Scale Multitask Learning},
  booktitle = {Proceedings of the 26th International Conference on Machine Learning (ICML)},
  pages = {1113--1120},
  year = {2009},
  timestamp = {Tue, 06 Nov 2018 00:00:00 +0100},
  biburl = {https://dblp.org/rec/conf/icml/WeinbergerDLSA09.bib},
  bibsource = {dblp computer science bibliography, https://dblp.org},
  doi = {10.1145/1553374.1553516},
  _bib2doi_selected = {dblp:/rec/conf/icml/WeinbergerDLSA09.bib},
  _bib2doi_confirmed = {true},
  _bib2doi_finished = {true},
}

@inproceedings{Stausholm21,
  author = {Stausholm, Nina Mesing},
  title = {Improved Differentially Private Euclidean Distance Approximation},
  booktitle = {Proceedings of the 40th ACM SIGMOD-SIGACT-SIGAI Symposium on Principles of Database Systems (PODS)},
  pages = {42--56},
  year = {2021},
  timestamp = {Thu, 19 Feb 2026 00:00:00 +0100},
  biburl = {https://dblp.org/rec/conf/pods/Stausholm21.bib},
  bibsource = {dblp computer science bibliography, https://dblp.org},
  doi = {10.1145/3452021.3458328},
  _bib2doi_selected = {dblp:/rec/conf/pods/Stausholm21.bib},
  _bib2doi_confirmed = {true},
}

@inproceedings{HogsgaardKLNS24,
  author = {H{\o}gsgaard, Mikael M{\o}ller and Kamma, Lior and Larsen, Kasper Green and Nelson, Jelani and Schwiegelshohn, Chris},
  title = {Sparse Dimensionality Reduction Revisited},
  booktitle = {Proceedings of the 41st International Conference on Machine Learning (ICML)},
  pages = {18454--18469},
  year = {2024},
  timestamp = {Mon, 09 Feb 2026 00:00:00 +0100},
  biburl = {https://dblp.org/rec/conf/icml/HogsgaardKLNS24.bib},
  bibsource = {dblp computer science bibliography, https://dblp.org},
  url = {https://proceedings.mlr.press/v235/hogsgaard24a.html},
  _bib2doi_selected = {dblp:/rec/conf/icml/HogsgaardKLNS24.bib},
  _bib2doi_confirmed = {true},
}

@article{ScalableDifferentiallyPrivateSketchesUnderContinualObservation,
  author = {Rayne Holland},
  title = {Scalable Differentially Private Sketches under Continual Observation},
  journal = {CoRR},
  volume = {abs/2507.03361},
  year = {2025},
  url = {https://doi.org/10.48550/arXiv.2507.03361},
  doi = {10.48550/arXiv.2507.03361},
  eprinttype = {arXiv},
  eprint = {2507.03361},
  timestamp = {Sat, 06 Sep 2025 01:00:00 +0200},
  biburl = {https://dblp.org/rec/journals/corr/abs-2507-03361.bib},
  bibsource = {dblp computer science bibliography, https://dblp.org},
  _bib2doi_selected = {dblp:/rec/journals/corr/abs-2507-03361.bib},
  _bib2doi_confirmed = {true},
}

@inproceedings{DifferentiallyPrivateContinualReleasesofStreamingFrequencyMomentEstimations,
  author = {Epasto, Alessandro and Mao, Jieming and Medina, Andres Munoz and Mirrokni, Vahab and Vassilvitskii, Sergei and Zhong, Peilin},
  title = {Differentially Private Continual Releases of Streaming Frequency Moment Estimations},
  booktitle = {14th Innovations in Theoretical Computer Science Conference (ITCS 2023)},
  pages = {48:1--48:24},
  series = {Leibniz International Proceedings in Informatics (LIPIcs)},
  isbn = {978-3-95977-263-1},
  issn = {1868-8969},
  year = {2023},
  volume = {251},
  editor = {Tauman Kalai, Yael},
  publisher = {Schloss Dagstuhl -- Leibniz-Zentrum f{\"u}r Informatik},
  address = {Dagstuhl, Germany},
  url = {https://drops.dagstuhl.de/entities/document/10.4230/LIPIcs.ITCS.2023.48},
  urn = {urn:nbn:de:0030-drops-175513},
  doi = {10.4230/LIPIcs.ITCS.2023.48},
  annote = {Keywords: Differential Privacy, Continual Release, Sliding Window, Streaming Algorithms, Distinct Elements, Frequency Moment Estimation},
  timestamp = {Thu, 02 Feb 2023 00:00:00 +0100},
  biburl = {https://dblp.org/rec/conf/innovations/EpastoMMMVZ23.bib},
  bibsource = {dblp computer science bibliography, https://dblp.org},
  _bib2doi_selected = {dblp:/rec/conf/innovations/EpastoMMMVZ23.bib},
  _bib2doi_confirmed = {true},
}

@inproceedings{Chan2017ModerateDimensions,
  author = {Chan, Timothy M.},
  title = {{Orthogonal Range Searching in Moderate Dimensions: k-d Trees and Range Trees Strike Back}},
  booktitle = {33rd International Symposium on Computational Geometry (SoCG 2017)},
  pages = {27:1--27:15},
  series = {Leibniz International Proceedings in Informatics (LIPIcs)},
  volume = {77},
  year = {2017},
  editor = {Aronov, Boris and Katz, Matthew J.},
  publisher = {Schloss Dagstuhl -- Leibniz-Zentrum f{\"u}r Informatik},
  address = {Dagstuhl, Germany},
  doi = {10.4230/LIPIcs.SoCG.2017.27},
  url = {https://drops.dagstuhl.de/entities/document/10.4230/LIPIcs.SoCG.2017.27},
  timestamp = {Wed, 16 Jun 2021 01:00:00 +0200},
  biburl = {https://dblp.org/rec/conf/compgeom/Chan17a.bib},
  bibsource = {dblp computer science bibliography, https://dblp.org},
  _bib2doi_selected = {dblp:/rec/conf/compgeom/Chan17a.bib},
  _bib2doi_confirmed = {true},
}

@inproceedings{Afshani2022HierarchicalCategories,
  author = {Afshani, Peyman and Killmann, Rasmus and Larsen, Kasper Green},
  title = {Hierarchical Categories in Colored Searching},
  booktitle = {33rd International Symposium on Algorithms and Computation (ISAAC 2022)},
  pages = {25:1--25:15},
  series = {Leibniz International Proceedings in Informatics (LIPIcs)},
  volume = {248},
  year = {2022},
  editor = {Bae, Sang Won and Park, Heejin},
  publisher = {Schloss Dagstuhl -- Leibniz-Zentrum f{\"u}r Informatik},
  address = {Dagstuhl, Germany},
  doi = {10.4230/LIPIcs.ISAAC.2022.25},
  url = {https://drops.dagstuhl.de/entities/document/10.4230/LIPIcs.ISAAC.2022.25},
  timestamp = {Tue, 14 Oct 2025 01:00:00 +0200},
  biburl = {https://dblp.org/rec/conf/isaac/AfshaniKL22.bib},
  bibsource = {dblp computer science bibliography, https://dblp.org},
  _bib2doi_selected = {dblp:/rec/conf/isaac/AfshaniKL22.bib},
  _bib2doi_confirmed = {true},
  _bib2doi_finished = {true},
}

@article{TheDiscreteGaussianforDifferentialPrivacy,
  title = {The Discrete {G}aussian for Differential Privacy},
  volume = {12},
  url = {https://journalprivacyconfidentiality.org/index.php/jpc/article/view/784},
  doi = {10.29012/jpc.784},
  number = {1},
  journal = {Journal of Privacy and Confidentiality},
  author = {Canonne, Cl{\'e}ment and Kamath, Gautam and Steinke, Thomas},
  year = {2022},
  month = {July},
  timestamp = {Mon, 28 Aug 2023 01:00:00 +0200},
  biburl = {https://dblp.org/rec/journals/jpc/CanonneKS22.bib},
  bibsource = {dblp computer science bibliography, https://dblp.org},
  _bib2doi_selected = {dblp:/rec/journals/jpc/CanonneKS22.bib},
  _bib2doi_confirmed = {true},
}

@inproceedings{SmoothBinaryMechanismForEfficientPrivateContinualObservation,
  author = {Andersson, Joel Daniel and Pagh, Rasmus},
  booktitle = {Advances in Neural Information Processing Systems},
  pages = {49133--49145},
  publisher = {Curran Associates, Inc.},
  title = {A Smooth Binary Mechanism for Efficient Private Continual Observation},
  url = {https://proceedings.neurips.cc/paper_files/paper/2023/file/99c41fb9fd53abfdd4a0259560ef1c9d-Paper-Conference.pdf},
  volume = {36},
  year = {2023},
  timestamp = {Fri, 01 Mar 2024 00:00:00 +0100},
  biburl = {https://dblp.org/rec/conf/nips/AnderssonP23.bib},
  bibsource = {dblp computer science bibliography, https://dblp.org},
  _bib2doi_selected = {dblp:/rec/conf/nips/AnderssonP23.bib},
  _bib2doi_confirmed = {true},
  _bib2doi_finished = {true},
}

@inproceedings{AlmostTightErrorBoundsOnDifferentiallyPrivateContinualCounting,
  author = {Henzinger, Monika and Upadhyay, Jalaj and Upadhyay, Sarvagya},
  title = {Almost Tight Error Bounds on Differentially Private Continual Counting},
  booktitle = {Proceedings of the 2023 Annual ACM-SIAM Symposium on Discrete Algorithms (SODA)},
  pages = {5003--5039},
  year = {2023},
  doi = {10.1137/1.9781611977554.ch183},
  url = {https://doi.org/10.1137/1.9781611977554.ch183},
  bibsource = {dblp computer science bibliography, https://dblp.org},
  timestamp = {Fri, 17 Feb 2023 00:00:00 +0100},
  biburl = {https://dblp.org/rec/conf/soda/HenzingerUU23.bib},
  _bib2doi_selected = {dblp:/rec/conf/soda/HenzingerUU23.bib},
  _bib2doi_confirmed = {true},
}

@article{ImprovedCountingUnderContinualObservationWithPureDifferentialPrivacy,
  author = {Andersson, Joel Daniel and Pagh, Rasmus and Torkamani, Sahel},
  title = {Improved Counting under Continual Observation with Pure Differential Privacy},
  journal = {CoRR},
  volume = {abs/2408.07021},
  year = {2024},
  doi = {10.48550/arXiv.2408.07021},
  url = {https://doi.org/10.48550/arXiv.2408.07021},
  eprinttype = {arXiv},
  eprint = {2408.07021},
  timestamp = {Mon, 23 Sep 2024 01:00:00 +0200},
  biburl = {https://dblp.org/rec/journals/corr/abs-2408-07021.bib},
  bibsource = {dblp computer science bibliography, https://dblp.org},
  _bib2doi_selected = {dblp:/rec/journals/corr/abs-2408-07021.bib},
  _bib2doi_confirmed = {true},
}

@inproceedings{PanPrivateAlgorithmsViaStatisticsOnSketches,
  author = {Mir, Darakhshan J. and Muthukrishnan, S. and Nikolov, Aleksandar and Wright, Rebecca N.},
  title = {Pan-private Algorithms via Statistics on Sketches},
  booktitle = {Proceedings of the Thirtieth ACM SIGMOD-SIGACT-SIGART Symposium on Principles of Database Systems (PODS)},
  pages = {37--48},
  year = {2011},
  doi = {10.1145/1989284.1989290},
  url = {https://doi.org/10.1145/1989284.1989290},
  timestamp = {Thu, 19 Feb 2026 00:00:00 +0100},
  biburl = {https://dblp.org/rec/conf/pods/MirMNW11.bib},
  bibsource = {dblp computer science bibliography, https://dblp.org},
  _bib2doi_selected = {dblp:/rec/conf/pods/MirMNW11.bib},
  _bib2doi_confirmed = {true},
}

@article{DifferentiallyPrivateMisraGries,
  author = {Lebeda, Christian Janos and Tet{\v e}k, Jakub},
  title = {Better Differentially Private Approximate Histograms and Heavy Hitters using the {Misra-Gries} Sketch},
  journal = {ACM Transactions on Database Systems},
  volume = {50},
  number = {3},
  pages = {9:1--9:26},
  year = {2025},
  doi = {10.1145/3716375},
  url = {https://doi.org/10.1145/3716375},
  timestamp = {Sat, 09 Aug 2025 01:00:00 +0200},
  biburl = {https://dblp.org/rec/journals/tods/LebedaT25.bib},
  bibsource = {dblp computer science bibliography, https://dblp.org},
  _bib2doi_selected = {dblp:/rec/journals/tods/LebedaT25.bib},
  _bib2doi_confirmed = {true},
}

@inproceedings{DifferentiallyPrivateWeightedSampling,
  author = {Cohen, Edith and Geri, Ofir and Sarl{\'o}s, Tam{\'a}s and Stemmer, Uri},
  title = {Differentially Private Weighted Sampling},
  booktitle = {Proceedings of The 24th International Conference on Artificial Intelligence and Statistics (AISTATS)},
  series = {Proceedings of Machine Learning Research},
  volume = {130},
  pages = {3068--3076},
  publisher = {PMLR},
  year = {2021},
  url = {https://proceedings.mlr.press/v130/cohen21b.html},
  timestamp = {Wed, 14 Apr 2021 01:00:00 +0200},
  biburl = {https://dblp.org/rec/conf/aistats/CohenGSS21.bib},
  bibsource = {dblp computer science bibliography, https://dblp.org},
  _bib2doi_selected = {dblp:/rec/conf/aistats/CohenGSS21.bib},
  _bib2doi_confirmed = {true},
  _bib2doi_finished = {true},
}

@inproceedings{SketchesBasedJoinSizeEstimationUnderLocalDifferentialPrivacy,
  author = {Zhang, Meifan and Liu, Xin and Yin, Lihua},
  title = {Sketches-Based Join Size Estimation Under Local Differential Privacy},
  booktitle = {Proceedings of the 40th IEEE International Conference on Data Engineering (ICDE)},
  pages = {1726--1738},
  publisher = {IEEE},
  year = {2024},
  doi = {10.1109/ICDE60146.2024.00140},
  url = {https://doi.org/10.1109/ICDE60146.2024.00140},
  timestamp = {Tue, 24 Mar 2026 00:00:00 +0100},
  biburl = {https://dblp.org/rec/conf/icde/ZhangLY24.bib},
  bibsource = {dblp computer science bibliography, https://dblp.org},
  _bib2doi_old_doi = {10.1109/ICDE60146.2024.00139},
  _bib2doi_selected = {dblp:/rec/conf/icde/ZhangLY24.bib},
  _bib2doi_confirmed = {true},
  _bib2doi_finished = {true},
}

@article{honaker2015efficient,
  title = {Efficient use of differentially private binary trees},
  author = {Honaker, James},
  journal = {Theory and Practice of Differential Privacy (TPDP)},
  volume = {2},
  pages = {26--27},
  year = {2015},
  _bib2doi_finished = {true},
}

@inproceedings{korolova2009releasing,
  title = {Releasing search queries and clicks privately},
  author = {Korolova, Aleksandra and Kenthapadi, Krishnaram and Mishra, Nina and Ntoulas, Alexandros},
  booktitle = {Proceedings of the 18th international conference on World Wide Web (WWW)},
  pages = {171--180},
  year = {2009},
  timestamp = {Tue, 06 Nov 2018 00:00:00 +0100},
  biburl = {https://dblp.org/rec/conf/www/KorolovaKMN09.bib},
  bibsource = {dblp computer science bibliography, https://dblp.org},
  doi = {10.1145/1526709.1526733},
  _bib2doi_selected = {dblp:/rec/conf/www/KorolovaKMN09.bib},
  _bib2doi_confirmed = {true},
}

@inproceedings{cormode2012differentially,
  title = {Differentially private summaries for sparse data},
  author = {Cormode, Graham and Procopiuc, Cecilia and Srivastava, Divesh and Tran, Thanh TL},
  booktitle = {Proceedings of the 15th International Conference on Database Theory},
  pages = {299--311},
  year = {2012},
  timestamp = {Sun, 06 Oct 2024 01:00:00 +0200},
  biburl = {https://dblp.org/rec/conf/icdt/CormodePST12.bib},
  bibsource = {dblp computer science bibliography, https://dblp.org},
  doi = {10.1145/2274576.2274608},
  _bib2doi_selected = {dblp:/rec/conf/icdt/CormodePST12.bib},
  _bib2doi_confirmed = {true},
  _bib2doi_finished = {true},
}

@article{aumuller2022representing,
  title = {Representing sparse vectors with differential privacy, low error, optimal space, and fast access},
  author = {Aum{\"u}ller, Martin and Lebeda, Christian Janos and Pagh, Rasmus},
  journal = {Journal of Privacy and Confidentiality},
  volume = {12},
  number = {2},
  year = {2022},
  timestamp = {Sat, 25 Feb 2023 00:00:00 +0100},
  biburl = {https://dblp.org/rec/journals/jpc/LebedaAP22.bib},
  bibsource = {dblp computer science bibliography, https://dblp.org},
  doi = {10.29012/jpc.809},
  _bib2doi_selected = {dblp:/rec/journals/jpc/LebedaAP22.bib},
  _bib2doi_confirmed = {true},
}

@article{li2015matrix,
  title = {The matrix mechanism: optimizing linear counting queries under differential privacy},
  author = {Li, Chao and Miklau, Gerome and Hay, Michael and McGregor, Andrew and Rastogi, Vibhor},
  journal = {The VLDB journal},
  volume = {24},
  number = {6},
  pages = {757--781},
  year = {2015},
  publisher = {Springer},
  timestamp = {Mon, 11 Sep 2017 01:00:00 +0200},
  biburl = {https://dblp.org/rec/journals/vldb/LiMHMR15.bib},
  bibsource = {dblp computer science bibliography, https://dblp.org},
  doi = {10.1007/s00778-015-0398-x},
  _bib2doi_selected = {dblp:/rec/journals/vldb/LiMHMR15.bib},
  _bib2doi_confirmed = {true},
}

@article{kotsogiannis2019privatesql,
  title = {{PrivateSQL}: a differentially private sql query engine},
  author = {Kotsogiannis, Ios and Tao, Yuchao and He, Xi and Fanaeepour, Maryam and Machanavajjhala, Ashwin and Hay, Michael and Miklau, Gerome},
  journal = {Proceedings of the VLDB Endowment},
  volume = {12},
  number = {11},
  pages = {1371--1384},
  year = {2019},
  publisher = {VLDB Endowment},
  timestamp = {Tue, 16 Aug 2022 01:00:00 +0200},
  biburl = {https://dblp.org/rec/journals/pvldb/KotsogiannisTHF19.bib},
  bibsource = {dblp computer science bibliography, https://dblp.org},
  doi = {10.14778/3342263.3342274},
  _bib2doi_selected = {dblp:/rec/journals/pvldb/KotsogiannisTHF19.bib},
  _bib2doi_confirmed = {true},
}

@inproceedings{Dvijotham2024efficient,
  author = {Krishnamurthy Dj Dvijotham and H. Brendan McMahan and Krishna Pillutla and Thomas Steinke and Abhradeep Thakurta},
  title = {Efficient and Near-Optimal Noise Generation for Streaming Differential Privacy},
  booktitle = {65th {IEEE} Annual Symposium on Foundations of Computer Science, {FOCS} 2024, Chicago, IL, USA, October 27-30, 2024},
  pages = {2306--2317},
  publisher = {{IEEE}},
  year = {2024},
  url = {https://doi.org/10.1109/FOCS61266.2024.00135},
  doi = {10.1109/FOCS61266.2024.00135},
  timestamp = {Mon, 09 Dec 2024 00:00:00 +0100},
  biburl = {https://dblp.org/rec/conf/focs/DvijothamMP0T24.bib},
  bibsource = {dblp computer science bibliography, https://dblp.org},
  _bib2doi_selected = {dblp:/rec/conf/focs/DvijothamMP0T24.bib},
  _bib2doi_confirmed = {true},
}

@inproceedings{andersson2025multiple,
author = {Andersson, Joel Daniel and Retschmeier, Lukas and Nelson, Boel and Pagh, Rasmus},
title = {Private lossless multiple release},
year = {2025},
publisher = {JMLR.org},
booktitle = {Proceedings of 42nd International Conference on Machine Learning (ICML)},
articleno = {59},
numpages = {20},
location = {Vancouver, Canada},
series = {ICML'25}
}

@article{Koufogiannis2017gradual, 
title={Gradual Release of Sensitive Data under Differential Privacy}, volume={7}, url={https://journalprivacyconfidentiality.org/index.php/jpc/article/view/649}, DOI={10.29012/jpc.v7i2.649}, number={2}, journal={Journal of Privacy and Confidentiality}, author={Koufogiannis, Fragkiskos and Han, Shuo and Pappas, George J.}, year={2017}, month={Jan.} }

\newpage
\appendix
\section{Preliminaries Omitted from Paper Body}
\label{app: Preliminaries}

{\bf Machine model.} Our data structure is presented for a variant of the standard Word RAM model with word size at least $\max(\log n,\log T, \log B)$.
In particular, this model allows rank-and-select queries on machine words to be answered in constant time.
To simplify the exposition, we assume that the Word RAM is augmented with the ability to represent and add real-valued Gaussians, and that fully random hash functions are available.
This This Gaussian-realization assumption can likely be removed by making use of discrete Gaussians~\citep{TheDiscreteGaussianforDifferentialPrivacy}.

\begin{definition}[Differential Privacy, \cite{TheAlgorithmicFoundationsofDifferentialPrivacy}]
    A randomized mechanism $\mathcal{M}: \mathcal{X}^c \rightarrow \mathcal{Y}$ is $(\epsilon, \delta)$-differentially private for $\epsilon > 0$ and $\delta \geq 0$ if for all subsets of outputs $\mathcal{S} \subseteq \mathcal{Y}$ and for all neighboring $x,x'\in \mathcal{X}^c$ such that $\norm{x-x'}_1 \leq 1$ it holds that:
    \begin{align*}
        \Pr[\mathcal{M}(x) \in \mathcal{S}] \leq \exp{\epsilon} \Pr[\mathcal{M}(x') \in \mathcal{S}] + \delta
    \end{align*}
    $\mathcal{M}$ satisfies approximate differential privacy when $\delta > 0$ and pure differential privacy when $\delta=0$.
    \label{def:Differential Privacy}
\end{definition}

When referring to $\rho$-zCDP (or zCDP), we refer to the definition of \cite{Bun2026ConcentratedDP}:
\begin{definition}[$\rho$-Zero-Concentrated Differential Privacy ($\rho$-zCDP), \cite{Bun2026ConcentratedDP}]
    A randomized mechanism $M: \mathcal{X}^c \rightarrow \mathcal{Y}$ is $(\varepsilon,\rho)$-zCDP if for all neighboring $x, x' \in \mathcal{X}^c$ such that $x$ and $x'$ only differ in one element and for all $\alpha \in \lr{1, \infty}$:
    \begin{align*}
        D_{\alpha}(M(x)||M(x')) \leq \varepsilon + \rho \alpha
    \end{align*}
    where $D_{\alpha}(M(x)||M(x'))$ is the $\alpha$-R\'enyi divergence between the distribution $M(x)$ and the distribution $M(x')$. \\
    A randomized mechanism $M: \mathcal{X}^c \rightarrow \mathcal{Y}$ is $\rho$-zCDP if it is $(0,\rho)$-zCDP.
\end{definition}

Moreover, we characterize the Gaussian Mechanism under zCDP:
\begin{definition}[Gaussian Mechanism, \cite{Bun2026ConcentratedDP}]
    A function $q: \mathcal{X}^n \rightarrow \R$ has sensitivity $\Delta_2$ if $\forall x,x' \in \mathcal{X}^n$ differing in a single entry it holds that $|q(x) - q(x')| \leq \Delta_2$. Let $q$ be a sensitivity-$\Delta_2$ query. 
    Then the mechanism $\M \; : \; \mathcal{X}^n \rightarrow \R$ that releases a sample from $\MN{q(x), \sigma^2}$ satisfies $\lr{\Delta_2^2/(2\sigma^2)}$-zCDP.
    \label{Lemma: Gaussian Mechanism for zCDP}
\end{definition}

\Cref{Lemma: Gaussian Mechanism for zCDP} generalizes to functions with values in $\R^k$, in which case the sensitivity bound is $\|q(x) - q(x')\|_2 \leq \Delta_2$ and the mechanism releases $\MN{q(x), \sigma^2 I_k}$, where $I_k$ is the identity matrix.

\section{Details for Section~\ref{sec:continual-counting}}\label{app:binary-mechanism-analysis}

Let $L$ be the linear map that takes a vector $y$ of dimension $2T-1$ indexed by the vertices in the complete binary tree such that for $t=0,\dots,T-1$ we have $(Ly)_t = \sum_{v\in P_t} y_v$.
As above we assign numerical indices $v = 0,\dots,T-1$ to the leaves and use indices $v = T,\dots,2T-2$ for the internal nodes.
Next, we describe a linear map $R$ such that the composition $LR$ is the all-ones lower triangular matrix, i.e., $(LRx)_t = A^{(t)}$ .
Define the \emph{sign}, denoted $\text{sign}(v)$, of a non-root node $v$ to be $-1$ if $v$ is a left child and $+1$ if $v$ is a right child, and let $\text{sib}(v)$ denote the sibling of a non-root node $v$.
\[
(Rx)_v = \begin{cases}
    \tfrac{1}{2} \sum_{t=0}^{T-1}x_t \text { if $v$ is the root node }\\
    \tfrac{1}{2} \, \text{sign}(v)\, \sum_{t: \text{sib}(v)\in P_t} x_t \text { if $v$ is a non-root internal node }\\
    \tfrac{1}{2} \, (x_v + \text{sign}(v)\, x_{\text{sib}(v)}) \text { if $v$ is a leaf }\\
\end{cases} \enspace .
\]
That is, the root node is associated with the sum of all inputs multiplied by $\tfrac{1}{2}$, each internal node $v$ is associated with the sum of all inputs corresponding to leaves in its \emph{sibling's} subtree multiplied by $\tfrac{1}{2} \, \text{sign}(v)$, and each leaf node is associated with $\tfrac{1}{2}$ times either a sum or a difference of sibling values.
By definition $(LRx)_t = (L(Rx))_t = \sum_{v\in P_t} (Rx)_v$.
We argue by induction on the depth that for every internal node $w$,
\begin{equation}\label{eq:pathsum}
\sum_{v\in P_w} (Rx)_{v} = \left(\sum_{t=0}^{\max\{i: w\in P_i\}} x_t\right) - \left(\tfrac{1}{2} \sum_{t: w\in P_t} x_t\right) \enspace .
\end{equation}
This is clearly true when $w$ is the root node. For the induction step we assume that the statement is true for $w$'s parent and see that adding $(Rx)_w$ exactly matches the claimed difference.
In turn, for a leaf node $t$ with parent $w$ we have, using (\ref{eq:pathsum}),
\begin{alignat*}{2}
\sum_{v\in P_t} (Rx)_{v}
& = \left(\sum_{v\in P_w} (Rx)_{v}\right) + \tfrac{1}{2}\,(x_t + \text{sign}(t)\,x_{\text{sib}(t)})\\
& = \left(\sum_{i=0}^{\max\{t: w\in P_t\}} x_i\right) - \left(\tfrac{1}{2} \sum_{i: w\in P_i} x_i \right) + \tfrac{1}{2} \, (x_t + \text{sign}(t)\, x_{\text{sib}(t)})\\
& = A^{(t)},
\end{alignat*}
where the last equality is obtained by checking the cases where $t$ is a left and a right child, respectively.
With the linear maps $L$ and $R$ defined, we see that the mechanism can be written in the form of a factorization mechanism:
\(
\tilde{A}^{(t)} = L (Rx + N)_t
\),
where $N \sim\mathcal{N}(0,\sigma^2)^{2T-1}$.
To analyze the privacy we first note that the $\ell_2$-sensitivity of $Rx$ is $\Delta_2(R) = \tfrac{1}{2}\sqrt{\log_2(T)+2}$.
This is because changing one input $x_i$ by at most $1$ changes exactly $(\log_2(T)+2)$ values $(Rx)_v$ by at most $1/2$, one changed node for each internal level and two at the leaf level.

We can now apply the standard zCDP bound for the Gaussian mechanism to $Rx$.
Releasing $Rx+N$ with $N\sim \mathcal{N}(0,\sigma^2)^{2T-1}$ satisfies $\rho$-zCDP with
$\rho \;=\; \Delta_2(R)^2/(2\sigma^2) \;=\; (\log_2(T)+2)/(8\sigma^2)$,
and since $L(Rx+N)$ is a (deterministic) post-processing of $Rx+N$, the released estimates satisfy the same privacy guarantee \citep{Bun2026ConcentratedDP}.
Equivalently, we can choose $\sigma^2 \;=\; (\log_2(T)+2)/(8\rho)$, resulting in
$\text{Var}[\tilde{A}^{(t)}] \;=\; (\log_2(T)+1)\sigma^2 < (\log_2(T)+2)^2/(8\rho)$.
This matches the asymptotic variance of the smooth Gaussian mechanism but reduces a lower-order term that is significant when $T$ is small.

\section{Details for Section \ref{sec:FastGaMe}}
\label{app:fastgame}

\subsection{Proof of Theorem~\ref{thm: Constant look-up time}}
\begin{proof}
Let $h=\log_2(T)$, and index levels on a root-to-leaf path by
$0,1,\ldots,h$, where level $0$ is the root.
For a leaf $t$, let $b(t)=\text{bin}(t-1)$ be its $h$-bit label, and let
$b(t)_{\leq j}$ denote its prefix of length $j$.
In the Gaussian binary-tree mechanism, each node $s\in\{0,1\}^{\leq h}$ has an independent noise
$\eta_s\sim \MN{0,\sigma^2}$, and the released noise at time $t$ is
\[
    \nu^{(t)}=\sum_{j=0}^h \eta_{b(t)_{\leq j}} .
\]
It is useful to write the partial sums on the active path as
\[
    S_t(j)=\sum_{r=0}^j \eta_{b(t)_{\leq r}} \qquad\text{for } 0\leq j\leq h .
\]
Then $\nu^{(t)}=S_t(h)$, $\text{Var}[\nu^{(t)}]=(h+1)\sigma^2$, and if $u$ is the longest common prefix of $b(t_1)$ and $b(t_2)$, then
\[
    \text{Cov}(\nu^{(t_1)},\nu^{(t_2)})=(|u|+1)\sigma^2 .
\]

The data structure is $DS=\langle l,w^{(l)},P^{(l)}\rangle$.
Here $l$ is the last queried time, $w^{(l)}\in\{0,1\}^{h+1}$ is a bit vector indexed by levels $0,\ldots,h$, and if $w^{(l)}_j=1$ then
$P^{(l)}_j$ is the already sampled value of $S_l(j)$ on the path to $b(l)$.
The bit vector fits in one machine word, and $P^{(l)}$ contains $h+1=\BU{\log T}$ words.
We maintain the invariant that, after processing any query sequence, the joint distribution of all released noise values together with all currently valid stored values $P^{(l)}_j$ is exactly the joint distribution of the corresponding variables in the Gaussian binary-tree mechanism.

Initially, no values are stored, and the invariant is trivial.
Consider a new query at time $t$.
If $l=t$, then the leaf has not changed, and we return the already stored value $P^{(l)}_h$.
Assume first that this is the first query.
We sample
\[
    P^{(t)}_0 \sim \MN{0,\sigma^2},
    \qquad
    P^{(t)}_h = P^{(t)}_0 + Z,
    \qquad
    Z\sim \MN{0,h\sigma^2},
\]
with $Z$ independent of $P^{(t)}_0$.
We set $w^{(t)}_0=w^{(t)}_h=1$, set $l=t$, and return $P^{(t)}_h$.
This is exactly the joint distribution of $(S_t(0),S_t(h))$, so the invariant holds.

It remains to handle the case where the previous queried time is $l\neq t$.
Let $u$ be the longest common prefix of $b(l)$ and $b(t)$, and set $m=|u|$.
The old and new active paths are identical up to level $m$ and disjoint below level $m$.
If $w^{(l)}_m=1$, the shared value $S_t(m)=S_l(m)$ is already stored as $P^{(l)}_m$.
Otherwise, let
\[
    a=\max\{j<m: w^{(l)}_j=1\},\qquad
    c=\min\{j>m: w^{(l)}_j=1\}.
\]
The sentinels $w^{(l)}_0=w^{(l)}_h=1$ ensure that these values exist whenever $w^{(l)}_m=0$.
By the Word RAM assumption, $a$ and $c$ are found in $\BU{1}$ worst-case time using rank/select.
Conditioned on $P^{(l)}_a=S_l(a)$ and $P^{(l)}_c=S_l(c)$, the Gaussian random walk value at level $m$ has the Brownian-bridge distribution
\[
    P^{(l)}_m \sim
    \MN{
        P^{(l)}_a + \frac{m-a}{c-a}\bigl(P^{(l)}_c-P^{(l)}_a\bigr),
        \frac{(m-a)(c-m)}{c-a}\sigma^2
    } .
\]
We sample from this conditional distribution, store the result in $P^{(l)}_m$, and set $w^{(l)}_m=1$.
This is the exact conditional law of the missing shared prefix sum in the Gaussian tree, so the induction invariant is preserved.

Below the node $u$, the new path to $b(t)$ uses tree nodes that have never appeared on any previous active path.
Indeed, at the first bit after $u$, the path to $b(l)$ goes left and the path to $b(t)$ goes right; since query times are increasing, every earlier queried leaf lies in the left subtree.
Thus, conditioned on $P^{(l)}_m$, the remaining sum on the new branch is independent of all previously released and stored values.
We sample
\[
    P^{(t)}_h = P^{(l)}_m + Z',
    \qquad
    Z'\sim \MN{0,(h-m)\sigma^2},
\]
with $Z'$ independent of the past, and release $\nu^{(t)}=P^{(t)}_h$.
Again, this is the exact conditional distribution of $S_t(h)$ given the variables already represented in the invariant.

Finally, we update the valid bits.
All valid levels $j \leq m$ remain valid because they lie on the shared prefix of the old and new paths; levels $j > m$ from the old path are no longer on the active path and are marked invalid, except that we set level $h$ valid for the newly sampled leaf value.
This is a constant number of word operations on $w$, and the only numerical samples drawn are a constant number of Gaussians.
Therefore, each query is processed in $\BU{1}$ worst-case time, and the structure uses $\BU{\log T}$ words.

By induction, the complete sequence of released values has the same joint distribution as in the Gaussian binary-tree mechanism.
\end{proof}

\newpage

\subsection{Pseudocode}
\label{Appendix: Pseudocode}

\begin{algorithm}[h!]
\caption{\textsc{NoiseUpdate}$(DS_i, t, \sigma^2, T)$}
\label{alg:lazy-noise}
\begin{algorithmic}[1]
\Require
  Data structure for coordinate $i$ at last query:
  $DS_i=(l_i, w^{(i,l_i)}, P^{(i)})$,
  where $h=\log_2(T)$, $w^{(i,l_i)}\in\{0,1\}^{h+1}$,
  and $P^{(i)}\in\mathbb{R}^{h+1}$ are indexed by levels $0,\ldots,h$.
\Ensure
  Updated $DS_i'=(l_i', w^{(i,t)}, P^{(i)})$ and noise value
  $\nu^{(t)}_i=P^{(i)}_{h}$.
\State {\bf Invariant:} If $w_j^{(i,l_i)} = 1$ then $P_j^{(i)}$ is the sampled
level-$j$ prefix sum on the active path to $\mathrm{bin}(l_i-1)$.

\State $h\gets \log_2(T)$ \Comment{levels are $0,\ldots,h$; leaf labels have $h$ bits}
\If{$l_i$ is \textbf{unset}} \Comment{first query to coordinate $i$}

    \State $w^{(i,t)} \gets 0^{h+1}$
    \State Sample $P^{(i)}_{0} \sim \mathcal{N}(0,\,\sigma^2)$
    \State Sample $Z \sim \mathcal{N}(0,\,h\sigma^2)$
    \State $P^{(i)}_{h} \gets P^{(i)}_{0}+Z$
    \State $w^{(i,t)}_0 \gets 1$, $w^{(i,t)}_h \gets 1$
    \Comment{invariant preserved since $P^{(i)}_{0}$ and $P^{(i)}_{h}$ have been set}
    \State $l_i' \gets t$
    \State \Return $(DS_i'=(l_i',w^{(i,t)},P^{(i)}),\ \nu^{(t)}_i=P^{(i)}_h)$
\EndIf
\If{$l_i = t$} \Comment{already queried at this time}
    \State \Return $(DS_i,\ \nu^{(t)}_i=P^{(i)}_h)$
\EndIf
\State $u \gets \mathrm{bin}(l_i-1)$,\quad $v \gets \mathrm{bin}(t-1)$
\State $m \gets |\mathrm{LCP}(u,v)|$ \Comment{old and new active paths agree through level $m$}
\State $w^{(i,t)} \gets w^{(i,l_i)}$
\If{$w^{(i,l_i)}_m = 0$} \Comment{the shared prefix value has not yet been sampled}
    \State $a \gets \text{pred}_{w^{(i,l_i)}}(m)$,\quad
           $c \gets \text{succ}_{w^{(i,l_i)}}(m)$
    \Comment{nearest valid levels strictly above/below $m$}
    \State $\mu \gets P^{(i)}_{a}
      + \frac{m-a}{c-a}\bigl(P^{(i)}_{c}-P^{(i)}_{a}\bigr)$
    \State $\tau^2 \gets \frac{(m-a)(c-m)}{c-a}\,\sigma^2$
    \State Sample $P^{(i)}_{m} \sim \mathcal{N}(\mu,\,\tau^2)$
      \Comment{Brownian-bridge step conditioning on endpoints}
    \State $w^{(i,t)}_m \gets 1$
      \Comment{shared prefix value is now valid}
\EndIf
\State Sample $Z' \sim \mathcal{N}(0,\,(h-m)\sigma^2)$
\State $P^{(i)}_{h} \gets P^{(i)}_{m}+Z'$
    \Comment{sample new independent tail from level $m$ to leaf $\mathrm{bin}(t-1)$}
\For{$j=m+1$ \textbf{to} $h-1$}\Comment{can be implemented in $\BU{1}$ time with a bit-mask}
    \State $w^{(i,t)}_j \gets 0$
\EndFor
\State $w^{(i,t)}_h \gets 1$
\Comment{invariant preserved since $P^{(i)}_{m}$ and $P^{(i)}_{h}$ have been set}
\State $l_i' \gets t$
\State \Return $(DS_i'=(l_i',w^{(i,t)},P^{(i)}),\ \nu^{(t)}_i=P^{(i)}_h)$
\end{algorithmic}
\end{algorithm}

\section{Details for Section \ref{sec:private-ortogonal-range-queries}} \label{app: Range Counting}

\subsection{Proof of Lemma~\ref{Lemma: Space of one instance of PCS}} \label{app: lemma space one PCS}
Before proving the lemma, we start by introducing a result of \cite{PCS}, which shows that the error of the output of a Private CountSketch (denoted $\Tilde{x}_a$ for any $a \in [n]$) can be bounded in a way that extends the analysis of regular CountSketch.
Recall that $\text{tail}_b(y)$ is the vector obtained from removing the $b$ largest coordinates in absolute value from a vector $y$.
Then in our notation:
\begin{lemma}[\cite{PCS}, Theorem 3.3]
    For every $\alpha \in [0,1]$ and every $a \in [n]$, the estimation error of Private CountSketch with $t$ repetitions, row size $b$ and noise from $\mathcal{N}(0,\sigma^2)$ satisfies:
    \begin{align*}
        \Pr\lrs{|\Tilde{x}_a - x_a| > \alpha\max\lrc{\varphi, \sigma}} < 2\exp\lr{-\BL{\alpha^2 t}}
    \end{align*}
    where $\varphi = \norm{\text{tail}_b(y)}/ \sqrt{b}$.
    \label{Lemma: Private CountSketch estimation}
\end{lemma}

\noindent
Based on this, we can prove \Cref{Lemma: Space of one instance of PCS}:
\begin{proof}
    Without loss of generality we can assume $E < \norm{y}_1$.
    Choosing \(
        b=\left\lceil \frac{\norm{y}_1}{2E}\right\rceil
    \),
    by \Cref{Lemma: Lower bound on tail},
    \begin{align*}
        \varphi
        = \frac{\norm{\text{tail}_b(y)}_2}{\sqrt{b}}
        \leq \frac{\norm{y}_1}{2b}
        \leq E.
    \end{align*}
    Since the lemma assumes $E\geq\sigma$, we have $\max\lrc{\varphi,\sigma}\leq E$.
    Applying \Cref{Lemma: Private CountSketch estimation} with $\alpha=1$:
    \begin{align*}
        \Pr\lrs{|\Tilde{x}_a-x_a|>E} < 2\exp\lr{-\BL{t}}.
    \end{align*}
    Setting $t=\T{\log(1/\beta)}$ makes this probability at most $\beta$.
    The row size is $b=\BU{1+\norm{y}_1/E}$, so the total space is
    \(s=tb=\BU{t(1+\norm{y}_1/E)}\)
    words.
\end{proof}

\subsection{Details and guarantees for Static Private Range Queries} \label{app: Static M}

\begin{figure}
\makebox[\linewidth]{
    \begin{tabular}{l|c|c|c}
         {\bf Reference} & {\bf Privacy} & {\bf Maximum error} & {\bf Space}\\ \hline\hline
         \citep{Xiao2011Wavelet} & $\varepsilon$-DP & $(\log B)^{1.5d+1}/\varepsilon$ & $B^d$ \\
\hline
         \citep{Chan2011ReleaseStatistics} & $\varepsilon$-DP & $(\log B)^{1.5d+1}/\varepsilon$ & $B^d$ \\
\hline
         \multirow{2}{*}{\citep{PureDPRecQueries}} & \multirow{2}{*}{$\varepsilon$-DP} & $(\log(B)+ (\log n)^{1.5d+2})/\varepsilon$ & $n^d$\\
\cline{3-4}
         &  & $(\log n)^{2d+1}\log(B)/\varepsilon$ & $n (\log n)^d$\\
\hline
         \multirow{2}{*}{\citep{OptimalPrivateHalfspaceCountingViaDiscrepancy} + \citep{FactorizationNormsAndHereditaryDiscrepancy}
        } & \multirow{2}{*}{$(\varepsilon,\tfrac{1}{100})$-DP} & \cellcolor{red!25} $\BL{(\log n)^{d-1}/\varepsilon}$ & \multirow{2}{*}{any}\\
\cline{3-3}
         &  & \cellcolor{red!25}$\BL{(\log B)^{d-1}/\varepsilon}$ & \\
        \hline
          \multirow{1}{*}{{\bf This paper}} &  $\varepsilon^2/2$-zCDP & $E\geq (\log B)^{d+2}/\varepsilon$ & $n (\log B)^{d+1}/ E$\\
          \cline{2-4}
          \hline
    \end{tabular}}
    \caption{Overview of results on static private orthogonal range counting for a set of $n$ points in $[B]^d$. %
    Lower bounds in terms of $n$ and $B$, for worst-case inputs and no restriction on the other parameter, are shown with red background, and hold in particular for $\varepsilon^2/2$-zCDP.
    For simplicity, we assume that both $d$ and the privacy parameter $\varepsilon$ are $\BU{1}$, and order-notation is suppressed.
    Maximum error upper bounds hold with probability $1-1/B$.
    All listed static data structures support an orthogonal range counting query in time $\BU{(\log B)^d}$ or $\BU{(\log B)^{d+1}}$.
    An arbitrary number of queries can be performed with query time $\BU{(\log B)^{d+1}}$.%
    }
    \Description{A table comparing privacy guarantees, maximum error, and space bounds for static private orthogonal range counting.}
    \label{figure:rangecounting-overview-static}
\end{figure}

We now consider the static mechanism $\M$ described in \Cref{sec:private-ortogonal-range-queries}.
An overview of results in the static setting and comparison with previous work is shown in \Cref{figure:rangecounting-overview-static}.
\begin{theorem}
    Let $B>1$ be a power of 2. For every constant integer $d < B$, $\beta \in (0,1/2)$, and $\rho > 0$,
    let $\tau=\log(d/\beta)+d\log\log B$ and let
    \(
        E = \BL{d\tau(\log B)^{d+1+\log_B(4/\beta)}/\sqrt{\rho}}
    \)
    denote some fixed estimation error.
    There is a mechanism that takes as input a static set of $n$ points in $[B]^d$, and outputs a data structure with the following properties:
    \begin{itemize}
        \item Satisfies $\rho$-zero-concentrated differential privacy.
        \item Space usage is $\BU{\frac{d n (\log B)^{d}\tau}{E}+1}$ words.
        \item Supports orthogonal range counting queries that
        \begin{itemize}
            \item return unbiased estimators of the true range count,
            \item have worst-case additive estimation error $\BU{E}$ with probability at least $1-\beta$, and
            \item can be answered in worst-case time $\BU{(2 \log B)^d\tau}$.
        \end{itemize}
    \end{itemize}
    \label{Theorem: Reduced space}
\end{theorem}

We summarize the mechanism $\M$ from \Cref{sec:private-ortogonal-range-queries} here:
The mechanism outputs $(\log_2 B)^d$ Private CountSketches, each representing point counts in a partition of $[B]^d$ into dyadic rectangles.
Each sketch is parameterized with variance parameter $\sigma^2 = t(\log B)^d/(2\rho)$, row size $b = \BU{n/\sigma}$, and $t$ repetitions (to be specified).
Using composition over all sketches this ensures $\rho$-zero-concentrated differential privacy.
Range counting queries are answered by summing at most $(2\log_2 B)^d$ estimates, at most $2^d$ from each sketch.
Let $\hat{X}_i$ denote the sum of estimates from sketch $i$, in some arbitrary numbering of the sketches.
Since Private CountSketch is unbiased, $\E[\hat{X}_i]$ equals the true number of points in the queried dyadic rectangles.
The following lemma bounds the error of a query.
\begin{lemma}\label{Lemma: Bernstein bound}
    Consider any specific range counting query and let $X_i = \hat{X}_i - \E[\hat{X}_i]$ be the signed estimation error from the $i$th Private CountSketch such that the total error is $|\sum_i X_i|$.
    Then for any choice of $c$, for sufficiently large $B$, row size $b = \BU{n/\sigma}$, and $t = \T{c\log B + d\log\log B}$ the total error satisfies:
    \begin{align}\label{eq:bernstein-bound}
        \Pr\lrs{\lra{\sum_i X_i} \geq \frac{t(\log B)^{c+d+1}}{\sqrt{\rho}}} \leq 2B^{-c} \enspace .
    \end{align}
\end{lemma}

\begin{proof}
    We define the following clipped version of $X_i$:
    \[Y_i = \begin{cases}
        -\sigma, \text{ if } X_i < -\sigma \\
        X_i, \text{ if $X_i \in [-\sigma,\sigma]$}\\
        \sigma, \text{ if } X_i > \sigma \\
    \end{cases}\]
    Let $K = t(\log B)^{c+d+1}/\sqrt{\rho}$ be the target error in (\ref{eq:bernstein-bound}).
    Since $\sum_i X_i = \sum_i Y_i$ unless one of the variables $X_i$ exceeds $\sigma$ in absolute value, we can bound the probability in (\ref{eq:bernstein-bound}) as follows:
    \begin{align*}
        \Pr\lrs{\lra{\sum_i X_i} \geq K}
        \leq \sum_i \Pr\lrs{|X_i| > \sigma }
        + \Pr\lrs{\lra{\sum_i Y_i} \geq K} \enspace .
    \end{align*}
    By \Cref{Lemma: Space of one instance of PCS}, making the constants in the bounds for row size $b$ and repetition number $t$ large enough we have, for $B$ sufficiently large,
    \begin{align*}
        \sum_i \Pr\lrs{|X_i| \geq \sigma }
        \leq (2\log_2 B)^d \exp(-\BL{t}) \leq B^{-c} \enspace .
    \end{align*}
    To bound the second term, we use Bernstein's inequality.
    We have $|Y_i|\leq \sigma$ by definition, and since $X_i$ is symmetric around 0 we have $\E[Y_i] = 0$.
    Bernstein's inequality states:
    \begin{align}\label{eq:bernstein}
        \Pr\lrs{\lra{\sum_i Y_i} \geq K}
        &\leq 2\exp\lr{
            \frac{-\frac{1}{2}K^2}{
                \sum_i\E\lrs{Y_i^2} + \frac{1}{3}\sigma K
            }
        } \enspace .
    \end{align}
    It remains to lower bound the exponent of (\ref{eq:bernstein}).
    Since $Y_i$ is $X_i$ clipped to $[-\sigma, \sigma]$, we have $\E[Y_i^2] \leq \E[X_i^2]$, so to get an upper bound for the denominator it suffices to upper bound the second moment of $X_i$.
    For a fixed sketch $i$, the range query uses at most $2^d$ estimates from that sketch.
    We can write their signed errors as $Z_{i,1},\ldots,Z_{i,q_i}$, where $q_i\leq 2^d$, so that $X_i=\sum_{j=1}^{q_i} Z_{i,j}$.
    By \Cref{Lemma: Space of one instance of PCS} and the choice $b=\BU{n/\sigma}$, each fixed estimate has second moment $\E[Z_{i,j}^2]=\BU{\sigma^2}$.
    By Cauchy--Schwarz,
    \[
        \E[X_i^2]
        =
        \E\left[\left(\sum_{j=1}^{q_i} Z_{i,j}\right)^2\right]
        \leq q_i \sum_{j=1}^{q_i}\E[Z_{i,j}^2]
        \leq \BU{4^d\sigma^2}.
    \]
    Summing over the $(\log_2 B)^d$ sketches and using
    $\sigma^2=t(\log B)^d/(2\rho)$ gives
    \[
        \sum_i \E[X_i^2]
        \leq \BU{4^d(\log B)^d\sigma^2}
        =
        \BU{4^d t(\log B)^{2d}/\rho} .
    \]
    The denominator of (\ref{eq:bernstein}) satisfies
    \begin{align*}
        \sum_i\E\lrs{Y_i^2} + \frac{1}{3}\sigma K
        &\leq
        \BU{4^d t(\log B)^{2d}/\rho}
        +
        \BU{t^{3/2}(\log B)^{c+\frac{3d}{2}+1}/\rho} .
    \end{align*}
    Since $d$ is constant and
    $t=\T{c\log B+d\log\log B}$, the hidden constants in
    $t$ can be chosen so that, for sufficiently large $B$,
    \[
        \frac{K^2}{
            \sum_i\E\lrs{Y_i^2} + \frac{1}{3}\sigma K
        }
        \geq 2c\ln B + 2\ln 2 .
    \]
    Therefore
    \(
        \Pr\lrs{\lra{\sum_i Y_i}\geq K}
        \leq B^{-c}
    \), finishing the proof.
\end{proof}

We are now ready to show \Cref{Theorem: Reduced space}.
\begin{proof}
    Choose $c=\log_B(4/\beta)$, so that $2B^{-c}=\beta/2$.
    Use the mechanism summarized above with
    \[
        t=\T{\max\lrc{\log(d/\beta), c\log B+d\log\log B}}=\T{\tau}
        \quad\text{and}\quad
        \sigma^2=t(\log_2 B)^d/(2\rho) \enspace .
    \]
    Each point update changes one counter in each repetition of each of the
    $(\log_2 B)^d$ dyadic sketches, so the squared $\ell_2$-sensitivity of the
    full released sketch vector is $t(\log_2 B)^d$.
    By \Cref{Lemma: Gaussian Mechanism for zCDP}, the above choice of
    $\sigma^2$ gives $\rho$-zCDP.

    Unbiasedness follows by linearity: Each Private CountSketch point estimate
    is unbiased, and each range answer is the sum of the relevant dyadic count
    estimates.
    \Cref{Lemma: Bernstein bound} controls the error when summing over the dyadic
    decomposition of a given query, giving error probability at most~$\beta/2$.
    Hence the range-query error is
    \[
        \BU{t(\log B)^{c+d+1}/\sqrt{\rho}}
        =
        \BU{\tau(\log B)^{d+1+\log_B(4/\beta)}/\sqrt{\rho}}.
    \]
    This is \(\BU{E}\) due to the lower bound on $E$.
    It remains to account for space and query time.
    When $E\geq n$ (using a privately estimated count) we use the constant-space zero estimator.
    Otherwise, each dyadic sketch stores $t$ rows of size
    $b=\BU{1+n/E}$ by \Cref{Lemma: Space of one instance of PCS}.
    Since there are $(\log_2 B)^d$ sketches, the total space is
    \[
        \BU{(\log B)^d t(1+n/E)}
        =
        \BU{\frac{d n(\log B)^d\tau}{E}+1} \enspace .
    \]
    A query probes at most $2^d$ entries in each of the $(\log_2 B)^d$ sketches,
    and each probe reads $t$ repetitions.
    This gives query time $\BU{(2\log B)^d\tau}$.
\end{proof}

\end{document}